\documentclass[journal, 12pt, draftclsnofoot, onecolumn]{IEEEtran}
\usepackage[T1]{fontenc}
\usepackage{url}
\ifCLASSINFOpdf
\else
\fi

\usepackage{amsmath}
\usepackage{algorithmicx}
\usepackage{algorithm}
\usepackage{graphicx}
\usepackage{mathtools}
\usepackage{textcomp}
\usepackage{amssymb}
\usepackage{bbm}
\usepackage{algpseudocode}
\usepackage{nccmath,textcomp}
\usepackage{eso-pic}
\usepackage{amsthm}
\usepackage{verbatim}
\newtheorem{lemma}{Lemma}
\newtheorem{theorem}{Theorem}
\newtheorem{assumption}{Assumption}

\newtheorem{proposition}{Proposition}
\newtheorem{definition}{Definition}
\newtheorem{remark}{Remark}

\makeatletter
\let\MYcaption\@makecaption
\makeatother

\usepackage[font=footnotesize]{subcaption}

\makeatletter
\let\@makecaption\MYcaption
\makeatother
\begin{document}
\bstctlcite{IEEEexample:BSTcontrol}
\title{Semi-Blind Post-Equalizer SINR Estimation and Dual CSI Feedback for Radar-Cellular Coexistence}
\author{Raghunandan M. Rao,~\IEEEmembership{Student Member,~IEEE,} Vuk Marojevic,~\IEEEmembership{Senior Member,~IEEE,} Jeffrey H. Reed,~\IEEEmembership{Fellow,~IEEE}
\thanks{Raghunandan M. Rao and Jeffrey H. Reed are with Wireless@VT, the Bradley Department of ECE, Virginia Tech, Blacksburg, VA, 24061, USA (email: \{raghumr, reedjh\}@vt.edu).}
\thanks{Vuk Marojevic is with the Department of ECE, Mississippi State University, Mississippi State, MS, 39762, USA
(e-mail: vuk.marojevic@ece.msstate.edu).}
\thanks{The support of the U.S. National Science Foundation (NSF) Grants CNS-1564148 and CNS-1642873 are gratefully acknowledged.} }

\maketitle
\vspace{-40pt}
\begin{abstract}
Current cellular systems use pilot-aided statistical-channel state information (S-CSI) estimation and limited feedback schemes to aid in link adaptation and scheduling decisions. However, in the presence of pulsed radar signals, pilot-aided S-CSI is inaccurate since interference statistics on pilot and non-pilot resources can be different. Moreover, the channel will be \textit{bimodal} as a result of the periodic interference. In this paper, we propose a max-min heuristic to estimate the post-equalizer SINR in the case of non-pilot pulsed radar interference, and characterize its distribution as a function of noise variance and interference power. We observe that the proposed heuristic incurs low computational complexity, and is robust beyond a certain SINR threshold for different modulation schemes, especially for QPSK. This enables us to develop a comprehensive \textit{semi-blind} framework to estimate the wideband SINR metric that is commonly used for S-CSI quantization in 3GPP Long-Term Evolution (LTE) and New Radio (NR) networks. Finally, we propose \textit{dual CSI feedback} for practical radar-cellular spectrum sharing, to enable accurate CSI acquisition in the \textit{bimodal channel}. We demonstrate significant improvements in throughput, block error rate and retransmission-induced latency for LTE-Advanced Pro when compared to conventional pilot-aided S-CSI estimation and limited feedback schemes.
\end{abstract}


\IEEEpeerreviewmaketitle
\vspace{-5pt}
\section{Introduction}
\IEEEPARstart{I}{n} order to mitigate severe spectrum shortage in sub-6 GHz bands and meet the exponentially increasing demand for user data, spectrum sharing has been proposed. In sub 6-GHz frequency bands, radar systems are the major primary consumers of spectrum, where most commercial cellular and wireless LAN (WLAN) systems currently operate. Spectrum sharing with radars is efficient because of its waveform characteristics, and sparse deployment. In particular, spectrum sharing with \textit{pulsed radar systems} is more desirable because of the interference-free time duration that can be leveraged for secondary user operations. 

In the United States, the Federal Communications Commission (FCC) has ratified the rules for radar-communications coexistence in the 3550-3650 MHz \cite{FCC_3point5_GHz_Rules} and 5 GHz \cite{FCC_5_GHz_FirstOrder} bands. More recently, the radar-incumbent 1.3 GHz \cite{WhiteHouse_1point3_GHz_Consideration} and 3450-3550 MHz \cite{NTIA_3450_to_3550} bands have also been identified for spectrum sharing. Due to these ongoing developments, cellular standardization has evolved into support for operation in unlicensed frequency bands, such as License Assisted Access (LAA) \cite{LAA_Kwon_ComMag_2017} and the Third Generation Partnership Project (3GPP) 5G New Radio-Unlicensed (5G NR-U) standards. 

In addition, vehicular communications are supported by cellular radio access technologies (RAT). More recently, 3GPP Release 14 introduced the cellular vehicle-to-everything (C-V2X) protocol \cite{3GPPLTE_TS36213_v14}, which can operate either in the 5.9 GHz band, or the cellular operator's licensed band \cite{Gaurang_Biplav_CV2X_IEEE_Acc_2019}. Therefore, C-V2X systems would also have to share spectrum with other wireless systems such as Wi-Fi and radar. In particular, high-powered radars operating in the 5 GHz Unlicensed National Information Infrastructure B (U-NII B) bands \cite{FCC_5_GHz_FirstOrder} can cause adjacent channel interference to C-V2X systems.

Often, cooperation between radar and cellular systems is (a) impractical in the case of outdated civilian radar systems, and (b) impossible with military radars due to security concerns. In addition, due to the rapid progress of cellular technology compared to that of radar systems, the burden of harmonious coexistence is usually placed on cellular systems, which is the premise for this paper.
\vspace{-9pt}
\subsection{Related Work}
Prior works have proposed harmonious radar-cellular coexistence mechanisms in different operational regimes using multi-antenana techniques, waveform optimization, and opportunistic spectrum access. Multi-antenna techniques exploit the spatial degrees of freedom to minimize mutual interference, and methods such as subspace projection \cite{Khawar_DySPAN_Spec_Share_2014},  \cite{Mahal_MIMOLTE_MIMORadar_Coexist_TAES_2017}, robust beamforming \cite{Liu_Robust_MIMO_BF_Rad_Cell_Coexist_2017}, and MIMO matrix completion \cite{Li_MIMOMC_Radar_TSP_2016} have been investigated in the past. These works assume the availability of accurate channel state information at the radar and/or the cellular system, which is often infeasible, especially in the case of spectrum sharing with military radars.

Radar waveform optimization approaches using mutual information (MI)-based metrics have been investigated in \cite{Tang_Li_Spectr_constr_Rad_Wfm_TSP_2019} to mitigate interference to secondary users. In addition, new multicarrier waveforms such as Precoded SUbcarrier Nulled-Orthogonal Frequency Division Multiplexing  (PSUN-OFDM) \cite{kim2016psun}, and FREquency SHift (FRESH)-filtered OFDM \cite{Carrick_Reed_FRESH_TAES_2019} have been proposed to improve their resilience to pulsed interference. Unfortunately, these waveforms require significant changes to existing radar systems and cellular standards, which makes their implementation infeasible in the near future. 

Opportunistic spectrum sharing approaches have also been studied in the context of spectrum sharing with a \textit{rotating radar} \cite{Rot_radar_cellular_JSAC_2012}, \cite{Khan_DaSilva_WCM_2016}, that leverages partial or complete information about the radar behavior to maximize spectral utilization in time/frequency/spatial dimensions. However these are not easily applicable to systems such as search-and-track radars. 

Numerical and experimental studies of underlay radar-LTE spectrum sharing scenarios \cite{ghorbanzadeh2016radar}, \cite{Reed_TDLTE_Rad_2point5_WCL_2016} have demonstrated that practical LTE deployments can operate with negligible degradation with an exclusion zone radius of tens of kilometers, which is significantly smaller than what is used in current deployments. However in these regimes, the pulsed radar intermittently impairs the cellular signal, disrupting data resources and critical control mechanisms of the cellular system. 

\begin{table}[t]
	\renewcommand{\arraystretch}{1.1}
	\caption{Simulation Parameters: Underlay Spectrum Sharing between an LFM Pulsed Radar and LTE-A Pro Downlink}
	\label{Table_Radar_LTE_coexist_HybSINR_DCFB}
	\centering
	\begin{tabular}{|l| l|}
		\hline
		\textbf{Parameter} & \textbf{Description}\\
		\hline
		3GPP Releases & 8 to 14 (LTE to LTE-A Pro)\\
		\hline 
		Center Frequency & $2$ GHz \\
		\hline
		System Bandwidth & $10 \text{ MHz}$ \\
		\hline
		Transmission Mode & TM 0 (SISO) from Port 0 \cite{sesia2011lte} \\
		\hline
		Small-scale Fading & Extended Pedestrian A (EPA)\\
		& Doppler frequency $f_d = 10$ Hz\\
		\hline 
		CSI feedback mode & Periodic and Wideband \\
		\hline
		CSI estimation interval \cite{sesia2011lte} & $10$ ms\\
		\hline
		CSI delay & $8$ ms \\
		\hline
		HARQ mode & Asynchronous and Non-Adaptive\\
		& with up to 4 retransmissions \\
		\hline
		Radar Pulse repetition Interval & $3.125$ ms\\
		\hline
		Radar pulse width $(T_\mathtt{pul})$ & $5\ \mu s$ \\
		\hline
		Radar relative carrier & $0$ Hz \\
		frequency offset $(\Delta f_r)$ & \\
		\hline
		Radar sweep frequency $(f_s)$ & $5$ MHz \\
		\hline
	\end{tabular} 
\end{table}

	\begin{figure*}[t]
	\centering
	\begin{subfigure}[t]{0.48\textwidth}
		\raggedleft
		\includegraphics[width=3.2in]{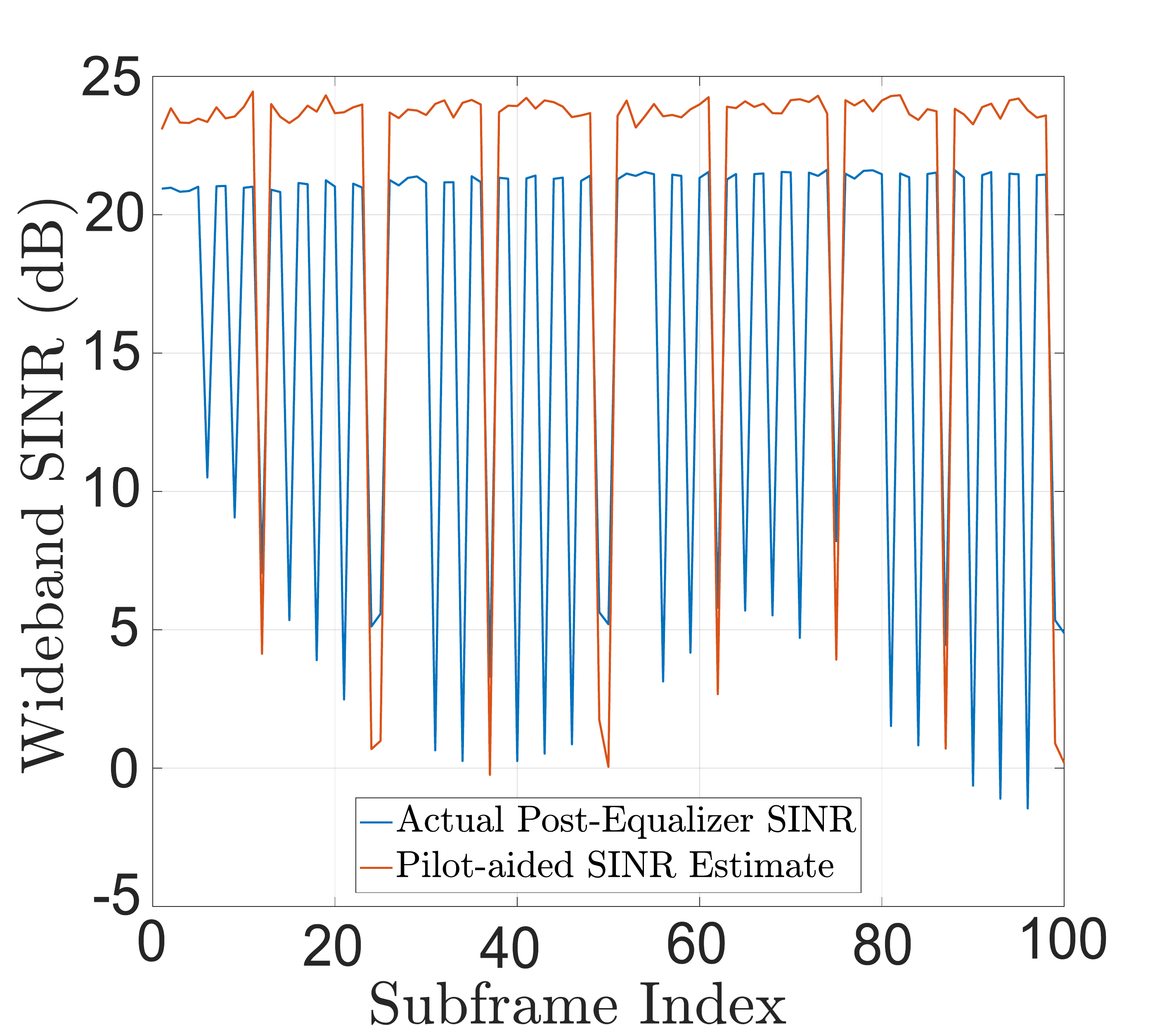}
		\caption{}
		\label{Fig_TVT_Act_vs_pilaid_SINR_frep_320}
	\end{subfigure}
	~
	\begin{subfigure}[t]{0.48\textwidth}
		\centering
		\includegraphics[width=3.4in]{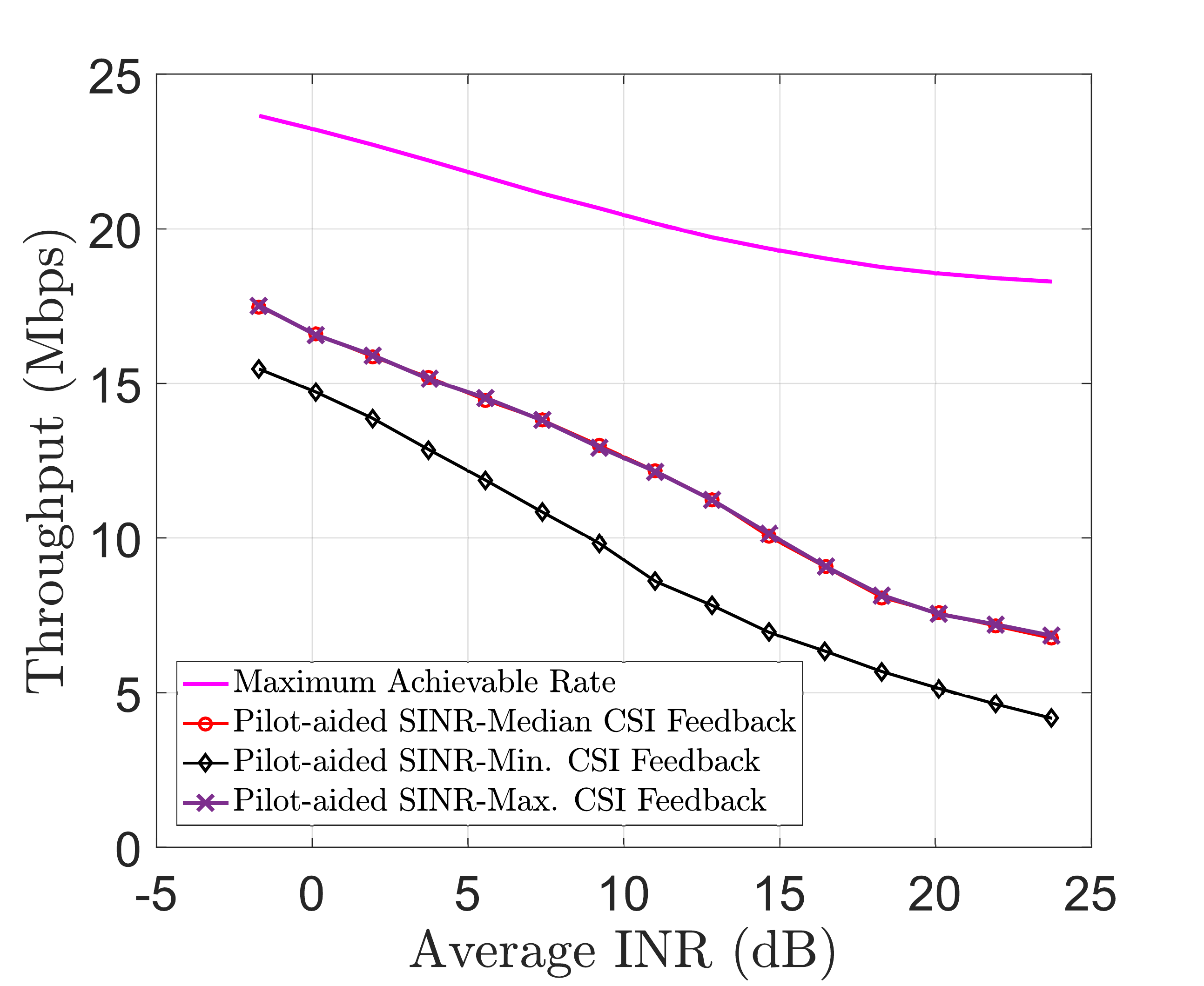}
		\caption{}
		\label{Fig_TVT_No_PropScheme_thpt_plots}
	\end{subfigure}
	~
	\begin{subfigure}[t]{0.48\textwidth}
		\centering
		\includegraphics[width=3.2in]{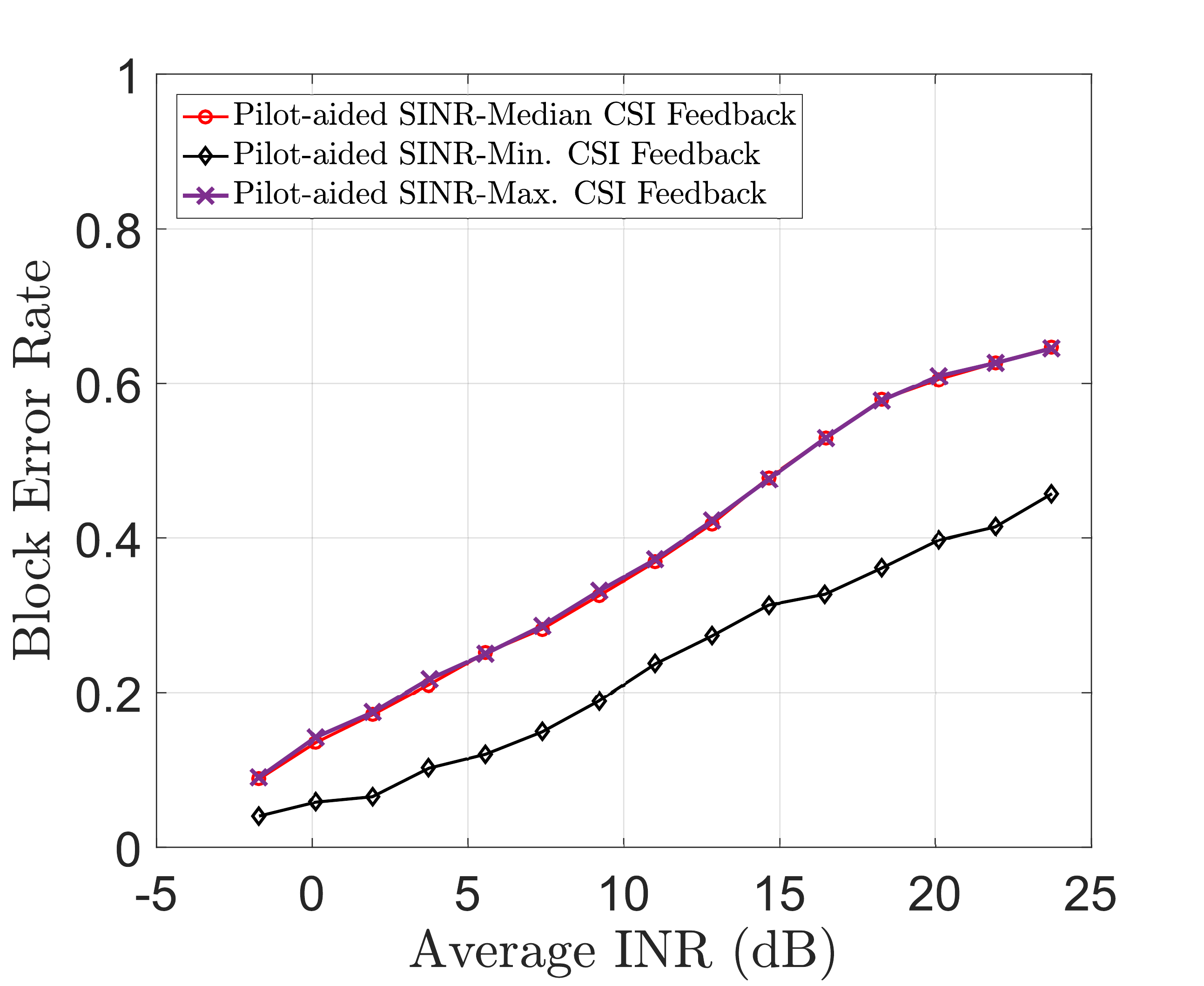}
		\caption{}
		\label{Fig_TVT_No_PropScheme_BLER_plots}
	\end{subfigure}
	\caption{Illustration of (a) inaccurate pilot-aided SINR estimates due to pulsed 	radar interference, and degradation of (b) throughput, and (c) block error rate (BLER) performance. The average SNR of the eNB-to-UE fading channel is $19.5$ dB.}
	\label{Fig_Illustrate_Problem}
\end{figure*}
\vspace{-10pt}
\subsection{Motivation}\label{Motivation_subsec}
Statistical-CSI (S-CSI) such as channel spatial covariance and \textit{post-equalizer SINR} \cite{rupp2016vienna} are important quantities, forming the basis of link adaptation and user scheduling schemes in modern wireless communication systems. LTE and NR systems use \textit{pilot-aided} S-CSI estimation schemes, and limited feedback mechanisms to balance link-level performance with feedback overhead \cite{sesia2011lte}, \cite{dahlman20185g}. Since pilots occupy a tiny fraction (at most 5\%) of time-frequency resources in cellular signals, pilot-aided S-CSI is accurate if interference and fading statistics are the same on pilot and non-pilot resources. While this is generally true in conventional cellular deployments, it does not hold in the presence of pulsed interference. 

We illustrate this using the example of a linear frequency modulated (LFM) pulsed radar coexisting with the LTE-A Pro downlink between a single evolved NodeB (eNB) and a single user equipment (UE). The system parameters shown in Table \ref{Table_Radar_LTE_coexist_HybSINR_DCFB}. The baseband transmitted waveform of the LFM radar is given by \cite{charles1993radar} 
	\begin{align}
	\label{Radar_waveforms}
	i_\mathtt{LFM}(t) & = \sqrt{P_\mathtt{rad}} e^{j \big( \tfrac{\pi f_s t}{T_\mathtt{pul}} + 2 \pi \Delta f_r \big) t} \text{ for } -\tfrac{T_\mathtt{pul}}{2}\leq t \leq \tfrac{T_\mathtt{pul}}{2},
	\end{align}
\noindent where $P_\mathtt{rad}$ is the radar transmitted power, $T_\mathtt{pul}$ is the radar pulse width, $f_s$ the sweep frequency, and $\Delta f_r$ the offset w.r.t. the center frequency of the cellular signal. 

Fig. \ref{Fig_TVT_Act_vs_pilaid_SINR_frep_320} illustrates the fundamental issue: Since pulsed radar interference is time-selective, absence of pilot interference results in inaccurate pilot-aided SINR estimates. As shown in Fig. \ref{Fig_TVT_No_PropScheme_thpt_plots} and \ref{Fig_TVT_No_PropScheme_BLER_plots}, this leads to degradation of throughput and block error rate performance in the case of commonly used limited feedback schemes: minimum, median and maximum CSI feedback (explained in Section \ref{Sec:Post-Eq-SINR}). Note that at the eNB, the criteria for choosing the modulation and coding scheme (MCS) is to maximize rate while satisfying $\text{BLER} \leq 0.1$ \cite{3GPPLTE_TS36213_v14}. The maximum achievable rate (pink curve) in Fig. \ref{Fig_TVT_No_PropScheme_thpt_plots} is the maximum rate achieved (using the maximum MCS) under the constraint that average $\text{BLER} = 0.1$.

Similar observations have been demonstrated in \cite{Rao_VTC_LTE_LinkAdapt_2019}, where \textit{non-pilot interference} was shown to significantly affect SINR estimates, resulting in the  degradation of link adaptation performance. In addition, results from \cite{Safavi_Roy_ICC_SINR_Nrwbnd_rad_2015} have shown lower link-level performance because of inaccurate CSI estimates in pulsed radar-LTE coexistence. Our prior work  \cite{Rao_Vuk_Reed_COMML_2020} rigorously proved that the S-CSI acquired for the \textit{interference channel} is inaccurate for a wide range of radar repetition intervals.

In such scenarios, blind SINR estimation methods need to be used since they do not rely on pilot signals. Prior works have investigated maximum likelihood (ML) \cite{Baumgartner_ML_2014}, \cite{Socheleau_ML_2008}, moment-based \cite{Chen_NDA_SINR_Estim_QAM_CL_2005}, and cyclostationary-based \cite{Hong_CycloSINR_2012} SINR estimation methods. However, the accuracy of ML and moment-based methods depend on the availability of accurate fading and interference statistics of the channel, which is often infeasible to acquire in real-time. For cyclostationarity-aided methods, short length of the cyclic prefix, unequal power allocation across subcarriers, and dependence of its accuracy on long-term averaging (for thousands of OFDM symbols) hinder their application to practical scenarios. Moreover, these methods do not estimate the \textit{post-equalizer} SINR\footnote{Post-equalizer SINR refers to the SINR of the received signal after channel estimation and equalization stages of the baseband receiver \cite{rupp2016vienna}, \cite{Kalyani_RatePredLTE_TWC_2016}.}, which is the metric used to aid scheduling decisions and link adaptation procedures in LTE and NR \cite{rupp2016vienna}. 

In addition to inaccurate SINR estimates in Fig. \ref{Fig_TVT_Act_vs_pilaid_SINR_frep_320}, we observe that the channel is \textit{bimodal}, due to periodic transitions between \textit{`interference-free'} and \textit{`interference-impaired'} states. Since limited feedback procedures in LTE and NR support single CSI feedback for a given frequency subband, it is fundamentally impossible to quantize the bimodal nature of the channel using a single value.
\vspace{-10pt}
\subsection{Contributions}
In this paper, we make the following contributions:
\begin{enumerate}
\item We present a robust max-min heuristic to estimate the post-equalizer SINR with low complexity, and characterize its distribution under a realistic tractable signal model for quadrature amplitude modulated (QAM) symbols. We analyze its accuracy and robustness, to demonstrate its applicability for radar-impaired OFDM symbols in practical spectrum sharing scenarios (section \ref{Low_Complex_Heur_SINR_Estimate}).

\item We propose a comprehensive framework to estimate the radar parameters, and combine pilot-aided as well as heuristic-aided SINR estimates to calculate the wideband post-equalizer SINR metric (section \ref{Hyb_SINR_Estim_Framwrk}).

\item We propose \textit{`dual CSI feedback'} as a simple extension to currently used limited CSI feedback mechanisms in cellular systems, to support CSI acquisition for \textit{`fading'} and \textit{`interference-impaired'} channel states (section \ref{sec:Dual-CSI-FB}). 
 
\item Using radar-LTE-A Pro spectrum sharing as an example, we demonstrate significant improvements in rate, BLER and retransmission-induced latency using our proposed framework, when compared to conventional \textit{pilot-aided SINR} and \textit{single CSI feedback} schemes (section \ref{sec:Dual-CSI-FB}).
\end{enumerate}
The rest of this paper is organized as follows. Section \ref{Sec:Post-Eq-SINR} provides the system model, and describes the basics of CSI estimation and limited feedback schemes used in LTE and NR. Section \ref{Low_Complex_Heur_SINR_Estimate} describes the max-min heuristic, derives its distribution under a tractable signal model, and analyzes its accuracy and robustness. Section \ref{Hyb_SINR_Estim_Framwrk} introduces the semi-blind SINR estimation framework, and evaluates its performance. Section \ref{sec:Dual-CSI-FB} develops the dual CSI feedback mechanism, discusses the incurred overhead, and demonstrates its effectiveness through link-level simulation results for radar-LTE-A Pro coexistence scenarios. Finally, section \ref{sec:Conclusions} concludes the paper, and discusses directions for future research.

\section{System Model and Preliminaries}\label{Sec:Post-Eq-SINR}
\subsection{Cellular Downlink Signal Model}
We consider an underlay radar-cellular spectrum sharing scenario, where the cellular downlink coexists with a \textit{wideband pulsed radar system}. For ease of exposition, we consider a single base station (with $N$ antenna ports) serving a single user (with $K$ antenna ports). The cellular downlink is OFDM-based with $N_\mathtt{sub}$ subcarriers, where data is transmitted in blocks composed of $T$ OFDM symbols. The received signal vector on the $k^{th}$ subcarrier of the $n^{th}$ OFDM symbol (referred to as a resource element (RE)) indexed by an ordered pair $(n,k)$), $\mathbf{z}_k [n] \in \mathbb{C}^{K}$, is given by
\begin{equation}
\label{MIMO_eqn}
\mathbf{z}_k [n]= \mathbf{H}_k [n] \mathbf{W}_k [n] \mathbf{x}_k [n] + \mathbf{h}_{r,k} [n] i_k [n]+ \mathbf{w}_k [n],
\end{equation}
where $\mathbf{H}_k [n] \in \mathbb{C}^{K \times N}$ is the downlink channel matrix, $\mathbf{W}_k [n] \in \mathbb{C}^{N \times L}$ the precoding matrix, and $L$ the data vector length. The transmitted symbol vector is chosen from $\mathbf{x}_k [n] \in \mathcal{X}^{L}$, where $\mathcal{X}$ is the set of symbols for the given modulation scheme. The noise vector is i.i.d. such that $\mathbf{w}_k [n] \sim \mathcal{CN} (0, \sigma_w^2 \mathbf{I}_K)$. After transmit beamforming, the radar-to-user channel vector on the $(n,k)^{th}$ resource elements (RE) is $\mathbf{h}_{r,k} [n] \in \mathbb{C}^{K}$, and the baseband-equivalent interference symbol is $i_k [n]$ such that $\mathbb{E}[\mathbf{h}_{r,k} [n] i_k [n]] = \mathbf{0}$ and $\mathbb{E}[\mathbf{h}_{r,k}[n] i_k[n] i^*_k [n] \mathbf{h}^{H}_{r,k} [n]] = \mathbf{R}_{\mathbf{I},k} [n]$. For ease of notation we suppress the RE index henceforth, while noting that the symbol on each RE is processed in a similar manner. 

If $\mathbf{\hat{H}}$ is the estimated channel matrix and $\hat{\sigma}^2_w$ the estimated noise variance, then the decoded data symbol $\mathbf{\hat{x}}$ using a minimum mean square error (MMSE) equalizer\footnote{In practical systems, other linear equalizers such as Zero-Forcing (ZF) or Regularized ZF are also commonly used to recover the data symbols.} is given by
\begin{align}
\label{Post_Eq_SINR_x_hat}
\mathbf{\hat{x}} & = (\mathbf{W}^H \mathbf{\hat{H}}^H \mathbf{\hat{H}W} + \hat{\sigma}^2_w \mathbf{I}_K)^{-1} \mathbf{W}^H \mathbf{\hat{H}}^H \mathbf{y}.
\end{align}
Defining $\mathbf{\hat{G}}_\mathtt{MMSE} \triangleq (\mathbf{W}^H \mathbf{\hat{H}}^H \mathbf{\hat{H}W} + \hat{\sigma}^2_w \mathbf{I}_L)^{-1} \mathbf{W}^H \mathbf{\hat{H}}^H$, the instantaneous SINR $\hat{\gamma}_l$ for the transmitted symbol on the $l^{th}$ antenna port ($1 \leq l \leq L$) is
\begin{equation}
\label{layer_posteq_SINR}
\hat{\gamma}_l = \frac{|x_l|^2}{\big|[ (\mathbf{\hat{G}}_\mathtt{MMSE} \mathbf{H W} - \mathbf{I}_L) \mathbf{x} + \mathbf{\hat{G}}_\mathtt{MMSE} (\mathbf{h}_r i + \mathbf{w})]_l\big|^2},
\end{equation}
where $[\mathbf{z}]_l$ denotes the $l^{th}$ element of $\mathbf{z}$. Since $\hat{\gamma}_l$ is calculated after baseband processing, it is termed as the \textit{post-equalizer/post-processing SINR}. This is used to calculate link quality metrics \cite{rupp2016vienna}, which subsequently aid in scheduling decisions and link adaptation schemes.
\vspace{-5pt}
\subsection{Pilot-Aided SINR Estimation and Wideband SINR Metrics}\label{Pil-Aided-SINR-Estim-Methd}
Typically, pilot signals are used both for channel estimation as well as for SINR estimation\footnote{3GPP Releases up to LTE-A Pro can use the common reference signal (CRS) and the demodulation reference signal (DMRS) to estimate the channel as well as the SINR. However, pilot signals such as the CSI reference signal (CSI-RS) can only be used to estimate the optimal precoder and SINR.}. In this work, we use the pilot-aided linear MMSE estimation method described in \cite{Mansour_NoiseVar_CQI_2015} assuming unit powered pilot symbols. 
For interference-free pilots in a MIMO transmission mode, the pilot-aided MMSE post-equalizer SINR estimate on the $l^{th}$ antenna port ($\hat{\gamma}_{p,l}$) is given by \cite{Li_Paul_Cioffi_MMSE_SINR_TIT_2006}
\begin{equation}
\label{SINR_pilot}
\hat{\gamma}_{p,l} = \frac{1}{\Big[ \frac{\mathbf{W}^H \mathbf{\hat{H}}^H \mathbf{\hat{H}W}}{\hat{\sigma}^2_w} + \mathbf{I}_L \Big]^{-1}_{l,l}} - 1,
\end{equation}
where $[\mathbf{X}]_{i,i}$ denotes the $i^{th}$ element on the main diagonal of matrix $\mathbf{X}$, and $p$ in the subscript of $\hat{\gamma}_{p,l}$ denotes that it is a pilot-aided SINR estimate. Since a data block comprises of contiguous time and frequency resource elements, a \textit{subband/wideband SINR metric} is often calculated to quantize the CSI. If the SINR estimate on the $(n,k)^{th}$ RE is $\hat{\gamma}[n,k]$, the wideband SINR is obtained using standard mapping functions such as effective exponential SINR mapping ($\hat{\gamma}_\mathtt{e}$) \cite{Donthi_Mehta_EESM_CQIFB_TWC_2011} and average SINR mapping ($\hat{\gamma}_\mathtt{a}$) \cite{Daniels_Heath_Sup_Learning_TVT_2010}, given by
\begin{align}
\label{SINR_Mapping_EESM_avg}
\hat{\gamma}_\mathtt{e} = \log \Big[  \sum_{(n,k) \in \mathcal{D}} \tfrac{e^{-\frac{\hat{\gamma}[n,k]}{\beta}}}{|\mathcal{D}|} \Big]^{-\beta } \text{ and } 
\hat{\gamma}_\mathtt{a} = \sum_{(n,k) \in \mathcal{D}} \tfrac{\hat{\gamma}[n,k]}{|\mathcal{D}|}
\end{align}
respectively, where $\mathcal{D}$ denotes the RE indices of data symbols in the cellular signal, and $\beta$ is a function of the modulation scheme \cite{Donthi_Mehta_EESM_CQIFB_TWC_2011}. 			
\vspace{-10pt}
\subsection{Link Adaptation Using Limited CSI Feedback}\label{Subsec_LinkAdapt_Lim_CSI_FB}
LTE and NR adapt the multi-antenna transmission mode (SISO/diversity/SU-MIMO/MU-MIMO), modulation format, and error control coding scheme, as a function of the channel fading and interference conditions. In order to limit the overhead while balancing performance, they support limited CSI feedback that is generally estimated over a finite estimation window, called the \textit{CSI estimation window}\footnote{The estimation window duration is chosen based on the rate at which the channel statistics vary, depending on user mobility. In typical cellular deployments, this interval ranges from tens to hundreds of milliseconds \cite{dahlman20185g}.}. The quantized CSI value consists of the following quantities:
\begin{enumerate} 
	\item Precoding Matrix Indicator (PMI):  an index of $\mathbf{W}_k \in \mathcal{W}$ chosen from a codebook $\mathcal{W}$ of predefined matrices.
	\item Rank Indicator (RI): the maximum rank supported on the downlink channel, which can be inferred from $\mathbf{W}_k$.
	\item Channel Quality Indicator (CQI): a 4-bit value representing the quantized subband/wideband post-equalizer SINR metric (\ref{SINR_Mapping_EESM_avg}) of the cellular signal.
\end{enumerate}
The CQI is mapped to a 5-bit modulation and coding scheme (MCS). In the LTE and NR PHY layer, decoding success and PHY layer metrics are characterized on units of data known as \textit{transport blocks}. For each transport block, the MCS denotes the most spectrally efficient scheme that simultaneously ensures that a maximum block error rate (BLER) is not exceeded on average. In addition, $L$ is equal to the number of transport blocks allotted to a single user, and $L \leq 2$ in LTE and NR even when the number of antenna ports $K \geq 2$ \cite{sesia2011lte}, \cite{dahlman20185g}. For ease of exposition, we refer to the wideband SINR metric of a data block as the \textit{post-equalizer SINR} henceforth. 
	
If $\hat{\gamma}[m]$ is the post-equalizer SINR calculated for the $m^{\text{th}}$ data block, then $CQI[m]=f(\hat{\gamma}[m]) \in \mathbb{N}$ is the corresponding CQI, where $f(\cdot)$ is a monotonically non-decreasing function of SINR. Considering a CSI estimation window of length $T_{CSI}$ data blocks, the wideband CQI measurements corresponding to the $T_{CSI}$ subframes are collectively represented by the vector $\mathbf{CQI} = \big[ CQI[0], CQI[1], \cdots, CQI[T_{CSI} - 1] \big] \in \mathbb{N}^{T_{CSI}}$.
In this work, we consider the following conventional CSI quantization and limited feedback schemes:
\begin{enumerate}
	\item minimum CSI feedback, where $\textmd{min}(\mathbf{CQI})$ is periodically fed back after every $T_{CSI}$ data blocks, and
	\item median CSI feedback, where $\textmd{med}(\mathbf{CQI})$ is periodically fed back after every $T_{CSI}$ data blocks.
	\item maximum CSI feedback, where $\textmd{max}(\mathbf{CQI})$ is periodically fed back after every $T_{CSI}$ data blocks.
\end{enumerate}
It is evident from Fig. \ref{Fig_Illustrate_Problem} that $\textmd{min}(\cdot)$, $\text{med}(\cdot)$ and $\text{max}(\cdot)$ quantization functions result in overoptimistic CQI values due to inaccurate pilot-aided SINR estimates. In addition, it is important to note that pilot-aided SINR estimates are accurate (a) in the absence of interference, and (b) when a pilot-bearing OFDM symbol is interference-impaired. Therefore, a key challenge is to accurately estimate the post-equalizer SINR with low computational complexity when pilot-resources are interference-free but data resources are not. In this work, we consider potential interference of pilots that are used to estimate the channel response as well as the SINR\footnote{In cellular standards up to LTE-A Pro, the same pilot signal is used for channel estimation as well as SINR estimation, such as the cell-specific reference signal (CRS). Other pilots such as Demodulation Reference Signals (DMRS) can also be used to estimate the SINR, \textit{conditioned on the precoding matrix} ($\mathbf{W}$) used \cite{dahlman20185g}.}. In the subsection below, we discuss the post-equalizer signal model of an interference-impaired non-pilot OFDM symbol, when the downlink channel is accurately estimated by interference-free pilot signals.

\begin{figure*}[!t]
	\centering
	\includegraphics[width=5in]{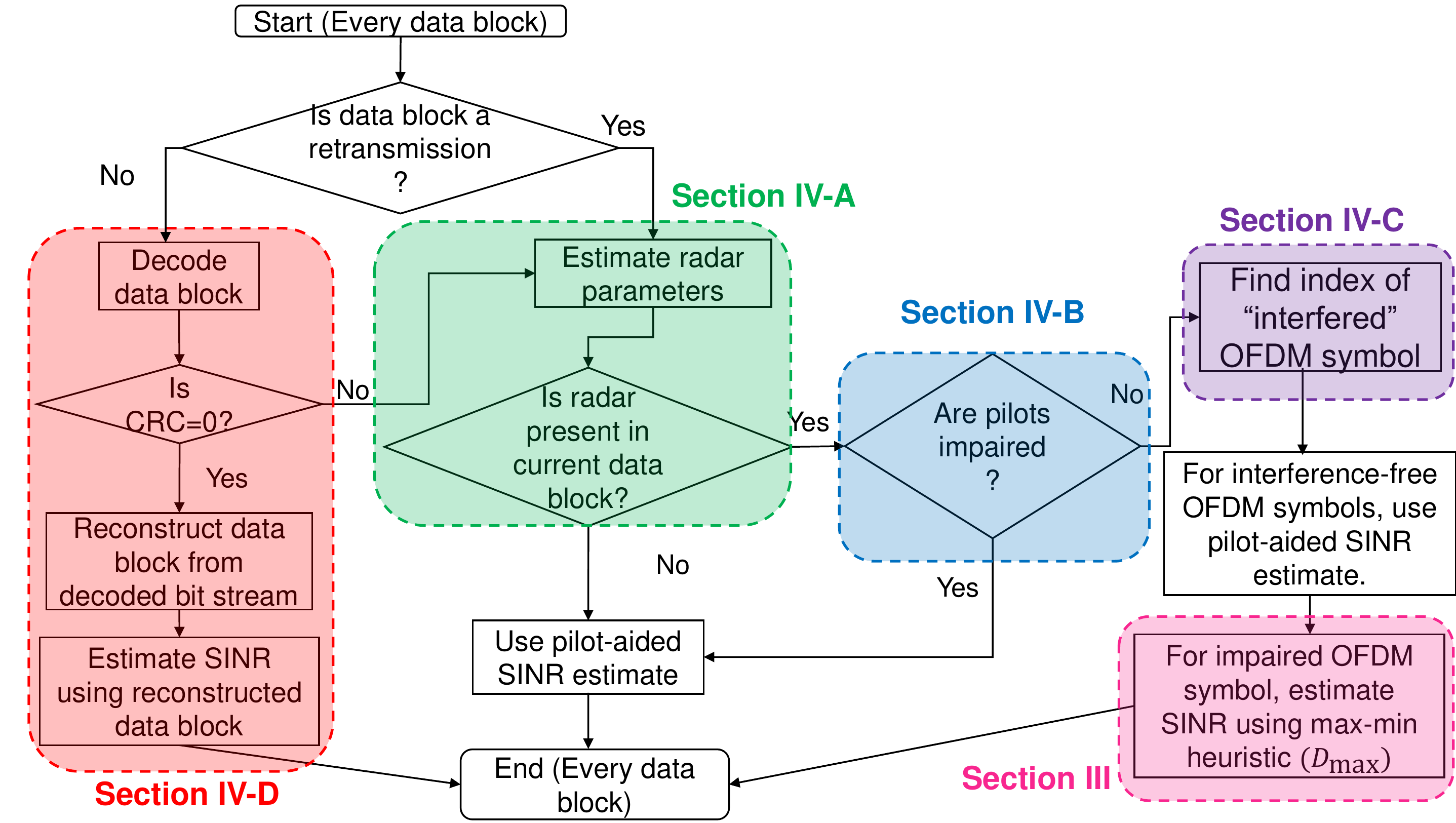}\\
	[-2ex]
	\caption{Flowchart of the hybrid SINR estimation framework for the cellular downlink in the presence of a pulsed radar. The SINR estimation algorithm is executed for every data block.}
	\label{SINR_Estim_Framework}
\end{figure*}
\vspace{-8pt}
\subsection{Baseband Equivalent Post-Processed Signal Model} \label{Post_Proc_Baseband_Eq_Signal_Model}
Using equation (\ref{Post_Eq_SINR_x_hat}), the resultant post-equalizer baseband signal on subcarrier $k$ will be
\begin{align}
\label{Equalized_signal}
\hat{\mathbf{x}}_k = \hat{\mathbf{G}}_{\mathtt{MMSE},k} \mathbf{H}_k \mathbf{W}_k \mathbf{x}_k + \hat{\mathbf{G}}_{\mathtt{MMSE},k} \mathbf{h}_{r,k} i_k + \hat{\mathbf{G}}_{\mathtt{MMSE},k} \mathbf{w}_k. 
\end{align}
To develop a tractable analytical model, we make the following assumptions.
\begin{assumption}\label{Assumpt_AWGN_equiv_model}
	In a coherence block of $K_{RB}$ subcarriers, the post-equalized signal on each antenna in the presence of accurate channel estimates can be written as
	\begin{align}
	\label{Post_processed_signal_model}
	y_k = x_k + \sqrt{P_{r,k}} e^{j \phi_k} + n_k,
	\end{align}
	where $P_{r,k}$ and $\phi_k$ is the post-equalizer interference power and phase, and $n_k \sim \mathcal{CN}(0, \sigma^2_n)$ is the additive white gaussian noise. The transmitted symbol $x_k \sim \mathtt{U}[\mathcal{X}]$, where $\mathcal{X}$ is the set of QAM symbols, and $\mathtt{U}[\cdot]$ denotes the uniform distribution.
\end{assumption}
\begin{assumption}\label{assump_amplitude}
	Interference power $P_{r,k}$ is constant in the coherence block $k \in \{k_0 + 1, k_0 + 2, \cdots, k_0 + K_{RB}\}$. 
\end{assumption}
\begin{assumption}\label{assump_phase}
	In the coherence block $k \in \{k_0 + 1, k_0 + 2, \cdots, k_0 + K_{RB}\}$  the interference phase is i.i.d. distributed as $\phi_k \sim \mathtt{U}(0, 2 \pi)$.
\end{assumption}
Assumption \ref{Assumpt_AWGN_equiv_model} approximates equation (\ref{Equalized_signal}) in a coherence block by an i\textit{nterference-impaired AWGN channel} on antenna port $l$ using $[\hat{\mathbf{x}}_k]_l = y_k, [(\hat{\mathbf{G}}_{\mathtt{MMSE},k} \mathbf{H}_k \mathbf{W}_k - \mathbf{I}_L)\mathbf{x}_k + \mathbf{\hat{G}}_{\mathtt{MMSE},k} \mathbf{h}_{r,k} i_k]_l = \sqrt{P_{r,k}} e^{j \phi_k}$, and $[\mathbf{\hat{G}}_{\mathtt{MMSE},k}  \mathbf{w}_k]_l = n_k$.

Assumption \ref{assump_amplitude} is accurate in a coherence block for LFM radar signals\footnote{This approximation is accurate in a coherence block of width $\sim 100$ kHz, in the case of a continuous-wave (CW) radar.} with $f_s T_\mathtt{pul} \gg 1$ where the radar spectrum is approximated by \cite{charles1993radar}
\begin{align}
\label{Radar_FT_approx}
I_\mathtt{LFM}(f) & \approx \sqrt{\frac{P_\mathtt{rad} T_\mathtt{pul}}{f_s}} e^{-j \big(\tfrac{\pi T_\mathtt{pul} (f - \Delta f_r)^2}{f_s} + \tfrac{\pi}{4} \big)} 
\end{align}
For an arbitrarily chosen contiguous subcarrier sequence $\{f_i\}$ for $i=1,2,\cdots,K_{RB}$ and $\Delta f_r = 0$, assumption \ref{assump_phase} approximates the sequence of square-law phase terms using
\begin{align}
\Big\{\frac{\pi T_\mathtt{pul} (k+i)^2 \Delta f^2}{f_s} \Big \} \stackrel{\text{i.i.d.}}{\sim} \mathtt{U}[0,2 \pi],  i=1,\cdots,K_{RB},
\end{align} 
after marginalization over a broad range of $0 \leq k \leq (N_\mathtt{sub} -(K_{RB} + 1))$,$f_s$, and $T_\mathtt{pul}$, where $\Delta f$ is the subcarrier spacing. Note that this approximation is used for ease of exposition, and the general form of the distribution is derived in the next section for scenarios when the phase offset of the radar interference is known.

\subsection{Post-Equalizer SINR Estimation Framework}	
To accurately estimate the wideband SINR metric in equation (\ref{SINR_Mapping_EESM_avg}), the receiver must be able to detect the presence of interference and localize its position in the time-frequency grid, so that the appropriate SINR estimate can be used for each RE. In this work, we propose a comprehensive framework to accurately estimate the post-equalizer SINR of a data block. Fig. \ref{SINR_Estim_Framework} shows the flowchart of the proposed framework, which is composed of the following key stages: 	
\begin{enumerate}
	\item Estimation of the radar repetition rate, which is used by the receiver to predict when radar interference will occur in the future.
	\item Detection of pulsed radar interference on pilot-bearing OFDM symbols, which is used by the receiver to determine the accuracy of pilot-aided SINR estimate for the interference channel.
	\item Detection of the contaminated OFDM symbol index. The receiver uses the \textit{max-min heuristic-aided SINR estimation method} only for the interference-impaired data-bearing OFDM symbol.
\end{enumerate}
In the following section, we characterize the properties of the proposed max-min heuristic that blindly estimates the post-equalizer SINR of a coherence block blindly in the presence of accurate downlink channel estimates, and the rest of the framework will be discussed in section \ref{Hyb_SINR_Estim_Framwrk}.

\section{Low Complexity Max-Min Heuristic to Estimate Post-Equalizer SINR}\label{Low_Complex_Heur_SINR_Estimate}
To estimate the \textit{post-equalizer interference and noise amplitude} in a coherence block of contiguous subcarrier indices $\{1, 2, \cdots, K_{RB}\}$, the heuristic $D_{\mathtt{max}}$ is defined as the maximum of the distance between a received symbol and its nearest neighboring constellation point, given by 
\begin{align}
\label{Prop_heuristic}
D_{\mathtt{max}} = \underset{k=1,2,\cdots,K_{RB}}{\max}\ \ \underset{x^{(j)} \in \mathcal{X}}{\min}\  \|y_k - x^{(j)} \|_2.
\end{align}
It is important to note that the additional complexity incurred is due to the $\max$ operation, since calculating the nearest neighbor distance is already a part of the downlink baseband processing chain in modern cellular systems. Therefore, the maximum \textit{minimum distance} calculated over a small coherence block of $K_{RB}$ REs incurs an additional computational complexity of $O(N_{RB} K_{RB})$, where $N_{RB}$ is the number of coherence blocks in the OFDM symbol. The cumulative distribution function (CDF) of $D_{\mathtt{max}}$ can be written as $F_{D_{\mathtt{max}}}(d) = \text{Pr} [D_{\mathtt{max}} \leq d],d \geq 0$. Defining the nearest-neighbor distance of the received symbol on the $l^{th}$ subcarrier as $D_{l} \triangleq \min_{x^{(j)} \in \mathcal{X}}\ \|y_l - x^{(j)} \|_2$, after defining $y \triangleq y_R + jy_I, n \triangleq n_R + j n_I$, and $x^{(j)} \triangleq x^{(j)}_R + jx^{(j)}_I \in \mathcal{X}$, the nearest neighbor distance can be simplified as
\begin{align}
\label{Expanding_gen_form_of_Dmin}
D  = & \big[(x_R - x^{(j)}_{R} + \sqrt{P_r} \cos \phi  + n_R)^2 + (x_I - x^{(j)}_{I} + \sqrt{P_r} \sin \phi  + n_I)^2 \big]^{1/2}.
\end{align} 

The following proposition denotes the marginal distribution of $D_l$ ($F_{D_l}(d)$) as a function of interference power $P_{r}$, and noise variance $\sigma^2_n$. 
\begin{proposition}
The CDF of $D$ can be written as
\begin{align}
\label{CDF_Dmin_Total_Prob_Thm}
F_{D_l}(d) = & \sum_{x \in \mathcal{X}} \int\displaylimits_{\mathcal{A}_{\Phi}} \int\displaylimits_{ \mathcal{A}_{n_R}} \int\displaylimits_{\mathcal{A}_{n_I}} \mathbbm{1}[D_l \leq d|x, n, \phi] p_{X}(x) f_{\Phi} (\phi) f_{N_R} (n_R) f_{N_I} (n_I) d \phi d n_R d n_I, d \geq 0,
\end{align}
where $\mathbbm{1}[\cdot]$ denotes the indicator function, $x \sim \mathtt{U}[\mathcal{X}]$, $p_X(x)$ is the probability mass function of $x \in \mathcal{X}$, $f_{\Phi} (\phi)$ is the density function of the radar phase $\phi$, $f_{N_R} (n_R)$  and $f_{N_I} (n_I)$ are the density functions of the real and imaginary components of noise, respectively. The corresponding integration regions are $\mathcal{A}_\Phi, \mathcal{A}_{N_R}$ and $\mathcal{A}_{N_I}$, respectively.

\end{proposition}
\begin{proof}
The event $\{D_l \leq d|x,n,\phi\}$ is represented by the indicator function $\mathbbm{1}[\cdot]$. Using the fact that the interference power $P_r$, phase $\phi$ and the real and imaginary components of noise are independent of each other, we obtain the desired result when the event of interest is integrated over the appropriate regions of $\phi, n_R$ and $n_I$.
\end{proof}
The marginal distribution of $D_\mathtt{max}$ is given in the following theorem.
\begin{theorem} \label{Thm_Marginal_CDF_D_max_exchangeable}
	If the interference phase relationship is known, and given by $\phi_i = h_i (\phi_1)$, where $\phi_1 \sim \mathtt{U}[0,2 \pi]$ is the phase of the first symbol in the coherence block and $i=2,3,\cdots,K_{RB}$, the marginal CDF of $D_\mathtt{max}$ is
	\begin{align*}
	F_{D_\mathtt{max}} (d) = \tfrac{1}{2 \pi |\mathcal{X}|^{K_{RB}}} \int\limits_{0}^{2 \pi} \prod_{l=1}^{K_{RB}} \Big[ \sum_{x_l \in \mathcal{X}} F_{D_l} (d|x_l, \phi_1) \Big] d \phi_1.
	\end{align*}
\end{theorem}
\begin{proof}
	The marginal CDF of $D_\mathtt{max}$ can be written as 
	\begin{align}
	\label{Marginal_CDF_first_princip}
	F_{D_\mathtt{max}} (d) = \int\limits_{0}^{2\pi} \sum\limits_{\mathbf{x} \in \mathcal{X}^{K_{RB}}} \mathbb{P}[\max(\mathbf{D}) \leq d|\mathbf{x}, \phi] p_\mathbf{X}(\mathbf{x}) f_\Phi(\phi) d \phi,
	\end{align}
	where $\mathbf{D}=[D_1,\cdots,D_{K_{RB}}]$ and $\mathbf{x} = [x_1, \cdots, x_{K_{RB}}] \sim \mathtt{U} [\mathcal{X}^{K_{RB}}]$. We have $\{ \max(\mathbf{D}) \leq d\} \Leftrightarrow \bigcap_{l=1}^{K_{RB}} \{D_l \leq d\}$. Since the phase relationship is deterministic when conditioned on $\phi_1$, the minimum distances ($D_l$) are conditionally independent. Marginalizing over the densities of $\mathbf{X}$ and $\Phi$, and simplifying equation (\ref{Marginal_CDF_first_princip}), we obtain the desired result.
\end{proof}

\begin{figure}[!t]
	\centering
	\includegraphics[width=6.0in]{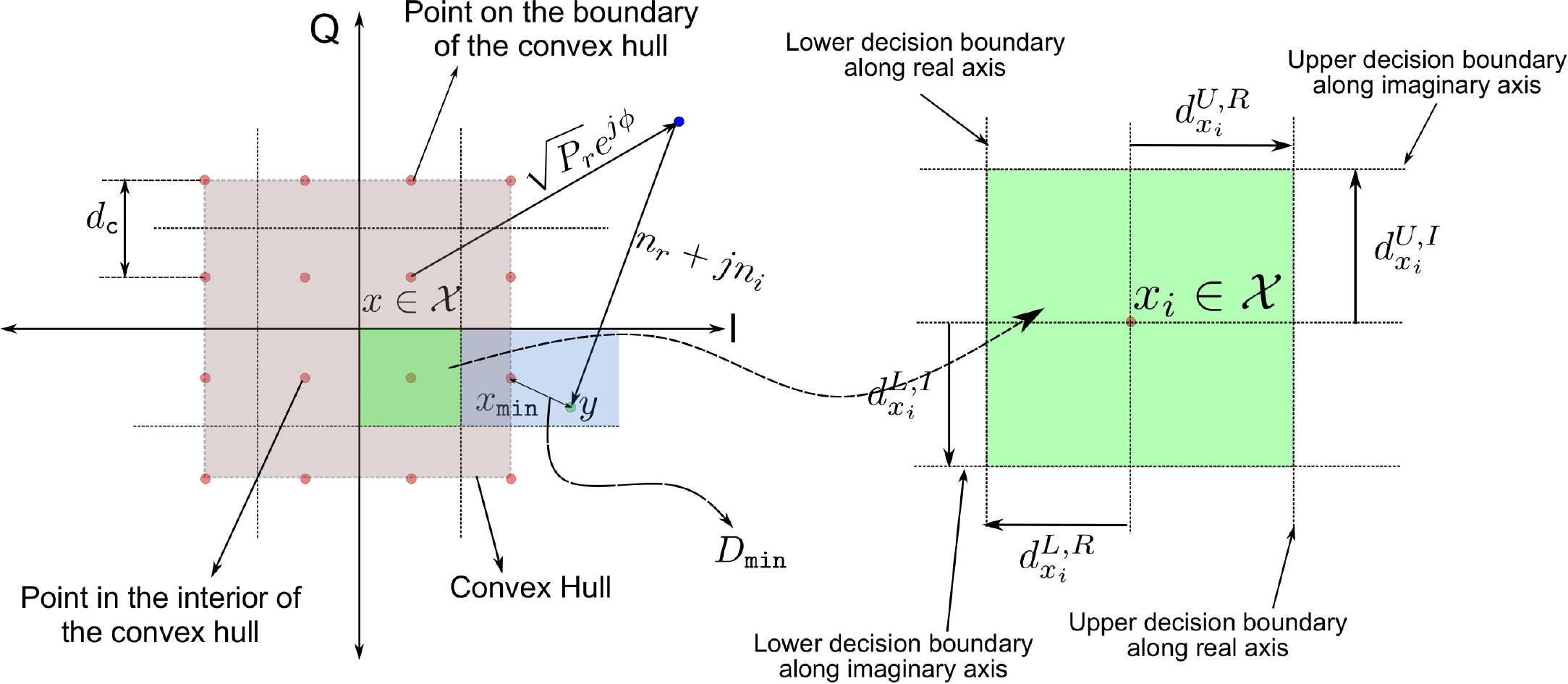}
	\caption{Illustration of transformation of the input to output symbols via interference and noise addition, and the resulting minimum distance $D_\mathtt{min} = \| y - x_\mathtt{min} \|_2$. The decision region of $x_\mathtt{min}$ is shaded in light blue. The lower figure shows the decision boundaries for $x_i \in  \mathcal{X}$ in the constellation diagram.}
	\label{Fig2_Input_to_Output_and_Dmin}
\end{figure}

In the case of the i.i.d. interference phase model, all the underlying random variables are independent of each other. Dropping the subcarrier index for notational simplicity, the marginal distribution is obtained using results from order statistics of i.i.d. random variables:
\begin{align}
\label{CDF_Dmax_iid_model}
F_{D_\mathtt{max}}(d) = [F_{D}(d)]^{K_{RB}}, d \geq 0.
\end{align}

\begin{table}[!t]
	\renewcommand{\arraystretch}{1.5}
	\caption{Decision region boundaries for constellation points in the first quadrant of 16-QAM}
	\label{Table_dist_to_decsn_bdries_16QAM}
	\centering
		\begin{tabular}{|c|c|c|c|c|}
			\hline
			$x_i \in \mathcal{X}_\mathtt{16QAM}$ & $d^{L,R}_{x_i}$ & $d^{U,R}_{x_i}$ & $d^{L,I}_{x_i}$ & $d^{U,I}_{x_i}$ \\
			\hline
			$\tfrac{1}{\sqrt{10}}(1+j)$ & $-\tfrac{1}{\sqrt{10}}$ & $\tfrac{1}{\sqrt{10}}$ & $-\tfrac{1}{\sqrt{10}}$ & $\tfrac{1}{\sqrt{10}}$ \\		
			\hline
			$\tfrac{1}{\sqrt{10}}(1+3j)$ & $-\tfrac{1}{\sqrt{10}}$ & $\tfrac{1}{\sqrt{10}}$ & $-\tfrac{1}{\sqrt{10}}$  & $\infty$\\		
			\hline
			$\tfrac{1}{\sqrt{10}}(3+j)$ & $-\tfrac{1}{\sqrt{10}}$ & $\infty$ & $-\tfrac{1}{\sqrt{10}}$ & $\tfrac{1}{\sqrt{10}}$ \\		
			\hline
			$\tfrac{1}{\sqrt{10}}(3+3j)$ & $-\tfrac{1}{\sqrt{10}}$  & $\infty$ & $-\tfrac{1}{\sqrt{10}}$ & $\infty$ \\
			\hline
		\end{tabular}
\end{table}

Based on their location w.r.t. the convex hull of the constellation, the transmitted symbols are classified as (a) interior points (denoted by set $\mathcal{X}_\mathtt{int}$) and (b) boundary points (denoted by set $\mathcal{X}_\mathtt{bnd}$), where $\mathcal{X}_\mathtt{bnd} \cap \mathcal{X}_\mathtt{int} = \emptyset$. 

Fig. \ref{Fig2_Input_to_Output_and_Dmin} illustrates the transformation of the baseband transmitted signal due to interference and noise. For a QAM scheme with average unit power per symbol, let the minimum distance between two points be $d_\mathtt{c}$. If the nearest neighbor of $y$ is $x^{(j)} \in \mathcal{X}_\mathtt{int}$, then $0 \leq D \leq \tfrac{d_\mathtt{c}}{\sqrt{2}}$. On the other hand, if $x^{(j)} \in \mathcal{X}_\mathtt{bnd}$, then $0 \leq D \leq \infty$. For each $x \in \mathcal{X}$, we define its decision region $\mathcal{A}_{x}$ given by 
\begin{align}
\label{Decision_regions_const_pts}
\mathcal{A}_{x} = & \{(z_x, z_y) |\Re(x) + d^{L,R}_{x} \leq z_x \leq \Re(x) + d^{U,R}_{x},  \Im(x) + d^{L,I}_{x} \leq z_y \leq \Im(x) + d^{U,I}_{x}\},
\end{align}
where $\Re(x)$ denotes the real part and $\Im(x)$ the imaginary part of complex scalar $x$. The decision region parameters $d^{L,R}, d^{L,I}, d^{U,R}, d^{U,I}$ for constellation points in the first quadrant of 16-QAM are shown in Table \ref{Table_dist_to_decsn_bdries_16QAM} and illustrated in the bottom portion of Fig. \ref{Fig2_Input_to_Output_and_Dmin}. In the following lemma, we derive the conditional distribution of $\{D|X, \Phi\}$. 

\begin{lemma} \label{Lemma_Cond_Dist_D_min_X}
The conditional distribution of $\{D|X, \Phi\}$ is given by 
\begin{equation}
\label{Conditional_CDF_D_min_X_Phi}
F_{D}(d | x, \phi) = \begin{cases}
\sum_{x^{(j)} \in \mathcal{X}} \Big[ 1 - Q_1 \Big( \frac{\sqrt{2}\nu_j}{\sigma_n}, \frac{\sqrt{2} d}{\sigma_n}\Big) \Big] & \text{if } 0 \leq d \leq \tfrac{d_\mathtt{c}}{2} \\
\sum_{x^{(j)} \in \mathcal{X}} \int_{0}^{d} \int_{\mathcal{A}_{\theta}(x^{(j)},z)} \frac{z}{\pi \sigma^2_n} e^{-\tfrac{z^2 + \nu^2_j + 2m_jz}{\sigma^2_n}} d\theta dz & \text{if } \tfrac{d_\mathtt{c}}{2} < d \leq \tfrac{d_\mathtt{c}}{\sqrt{2}} \\
F_D\big( \tfrac{d_\mathtt{c}}{\sqrt{2}} \big| x, \phi \big) + \sum_{x^{(j)} \in \mathcal{X}_\mathtt{bnd}} \int_{0}^{d} \int_{\mathcal{A}_{\theta}(x^{(j)},z)} \frac{z}{\pi \sigma^2_n} e^{-\tfrac{z^2 + \nu^2_j + 2m_jz}{\sigma^2_n}} d\theta dz & \text{otherwise}. 
\end{cases}
\end{equation}
where $m_j = m_{R,j} \cos \theta + m_{I,j} \sin \theta$, $\nu_j = (m^2_{R,j} + m^2_{I,j})^{1/2}$, $m_{R,j} = x_R - x^{(j)}_R + \sqrt{P_r} \cos \phi$, $m_{I,j} = x_I - x^{(j)}_I + \sqrt{P_r} \sin \phi$, and $Q_M(a,b)$ is the Marcum Q-function with parameters $M, a$ and $b$ \cite{abramowitz2012handbook}. The region of integration for $\Theta$ is given by $\mathcal{A}_{\theta}(x^{(j)}, z) = \big\{ \theta \big| d^{L,R}_{x^{(j)}} \leq z \cos \theta \leq d^{U,R}_{x^{(j)}}  , d^{L,I}_{x^{(j)}} \leq z \sin \theta \leq d^{U,I}_{x^{(j)}}, 0 \leq z \leq d \big\}$.
\end{lemma}
\begin{proof}
Refer Appendix \ref{Appndx_A}. 
\end{proof}
By marginalizing $\{D|X, \Phi\}$ over $\{X, \Phi\}$, the distribution of $D$ is given by
\begin{align}
	 \label{Marginal_CDF_D_min}
	 F_{D}(d) = \begin{cases}
	 \frac{1}{2 \pi |\mathcal{X}|} \sum\limits_{x \in \mathcal{X}} \sum\limits_{x^{(j)} \in \mathcal{X}} \int\limits_{0}^{2\pi} \Big[ 1 - Q_1 \Big( \frac{\sqrt{2} \nu_j}{\sigma_n}, \frac{\sqrt{2} d}{\sigma_n} \Big) \Big] d \phi, & \text{if } 0 \leq d \leq \tfrac{d_\mathtt{c}}{2}, \\
	 F_{D} \big(\tfrac{d_\mathtt{c}}{2} \big) + \sum\limits_{x \in \mathcal{X}} \sum\limits_{x^{(j)} \in \mathcal{X}} \int\limits_{0}^{2 \pi} \int\limits_{\tfrac{d_\mathtt{c}}{2}}^{d} \int\limits_{\mathcal{A}_{\theta}(x_j,z)}  \tfrac{z}{2\pi^2 |\mathcal{X}| \sigma^2_n} e^{-\tfrac{z^2 + \nu^2 + 2m_j z}{\sigma^2_n}} dz d \theta d \phi & \text{if } \tfrac{d_\mathtt{c}}{2} \leq d \leq \tfrac{d_\mathtt{c}}{\sqrt{2}} \\
	 F_{D} \big( \tfrac{d_\mathtt{c}}{\sqrt{2}} \big) + \sum\limits_{x \in \mathcal{X}} \sum\limits_{x^{(j)} \in \mathcal{X}_\mathtt{bnd}} \int\limits_{0}^{2 \pi} \int\limits_{\tfrac{d_\mathtt{c}}{\sqrt{2}}}^{d} \int\limits_{\mathcal{A}_{\theta}(x_j,z)} \tfrac{z}{2 \pi^2 |\mathcal{X}| \sigma^2_n} e^{-\tfrac{z^2 + \nu^2 + 2 m_j z}{\sigma^2_n}} dz d \theta d \phi & \text{otherwise}.
	 \end{cases}
\end{align}
Using it in (\ref{CDF_Dmax_iid_model}), we obtain the distribution of $D_\mathtt{max}$ under the i.i.d. interferer phase model. To characterize the robustness and accuracy of the interference-plus-noise power estimate, we define the following metrics. 

\begin{definition}
	Overestimation probability, defined as $P_\mathtt{overest} (P_r, \sigma^2_n) = \mathbb{P} [D_\mathtt{max} \geq \sqrt{P_r + \sigma^2_n}]$, is the probability that $D_\mathtt{max}$ overestimates the interference-plus-noise compared to the average interference-plus-noise power. 
\end{definition}
\begin{definition}
	Accuracy, defined as $\mathbb{P}\Big[\Big|\log_{10} \Big(\frac{D^2_\mathtt{max}}{P_r + \sigma^2_n} \Big) \Big| \leq \delta \Big]$, is the probability that the estimate of interference-plus-noise-power lies within a range of $\pm \delta$ (dB) of the actual value.
\end{definition}
A higher overestimation probability implies a more robust SINR estimate. As we will see in section \ref{Sec:Num-Results}, SINR estimation using the  proposed heuristic results in robust link adaptation in the presence of pulsed radar interference.
\vspace{-8pt}
\subsection{Numerical Results}
Fig. \ref{Fig_CDF_Heur_Theor_Sim_16QAM} shows the theoretical and simulated distributions of $D_\mathtt{max}$ for different values of $P_r$ and $\sigma^2_n$, and 16-QAM modulated data symbols with a coherence block length of $K_{RB}=12$ are used. We observe that there is very good agreement between the theoretical and numerical results, validating the accuracy of equations (\ref{CDF_Dmin_Total_Prob_Thm})-(\ref{Marginal_CDF_D_min}). In order to study the estimation accuracy, mismatch in interference-plus-noise of the heuristic compared to that of the average interference-plus-noise power is plotted in Fig. \ref{Fig_Heur_Ratio_Performance_16QAM} for 16-QAM symbols. We observe that the SINR mismatch in the interference-impaired OFDM symbol is within $\pm 5$ dB for more than $90\%$ of the range of typical SINR values ($-5 \text{ to } 30$ dB) encountered in cellular communications. However, in typical scenarios where at most a single radar pulse impacts a data block, mismatch in the wideband SINR metric ($\gamma_\mathtt{avg}$/$\gamma_\mathtt{eesm}$) due to the robust heuristic will be partially mitigated by the availability of accurate pilot-aided SINR estimates for \textit{interference-free OFDM symbols}, as discussed in the following section.

Fig. \ref{Fig_Robustness_all} shows the probability of overestimation as a function of $(P_r,\sigma^2_n)$ for different QAM schemes. We observe that the robustness of the heuristic decreases when the modulation order increases from QPSK to 64-QAM, and that $P_\mathtt{overest}(P_r, \sigma^2_n) \geq 0.9$ for QPSK. The reason for this trend can be intuitively explained by considering the following. 

If the transmitted symbol is $x \in \mathcal{X}_\mathtt{bnd}$, the received symbol $y$ will have a high probability of lying outside the convex hull. In this case, if the nearest neighbor lies on the convex hull and is $x' \in \mathcal{X}_\mathtt{bnd}$, $\| y - x\|_2 $ and $\| y - x' \|_2$ will have the same order of magnitude. In other words, the penalty due to nearest-neighbor association (i.e. $x'$ instead of $x$) will be minimal. 

On the other hand, for any $x \in \mathcal{X}$, if the nearest neighbor lies within the convex hull i.e. $x' \in \mathcal{X}_\mathtt{int}$, then a constellation with a higher minimum distance ($d_\mathtt{c}$) will be more robust. Since $\| y - x' \| \leq d_\mathtt{c}/2$, constellations with a higher $d_\mathtt{c}$ intrinsically has a higher probability of overestimating $\| y - x \|_2$. 

Since QPSK (a) has the highest minimum distance of $d_{c, \mathtt{QPSK}} = 1/\sqrt{2}$, and (b) has all points lying on the convex hull, the max-min heuristic is more robust when compared to that for 16-QAM and 64-QAM.

\begin{figure}[!htbp]
	\centering
	\begin{subfigure}[t]{0.48\textwidth}
		\raggedleft
		\includegraphics[width=3.2in]{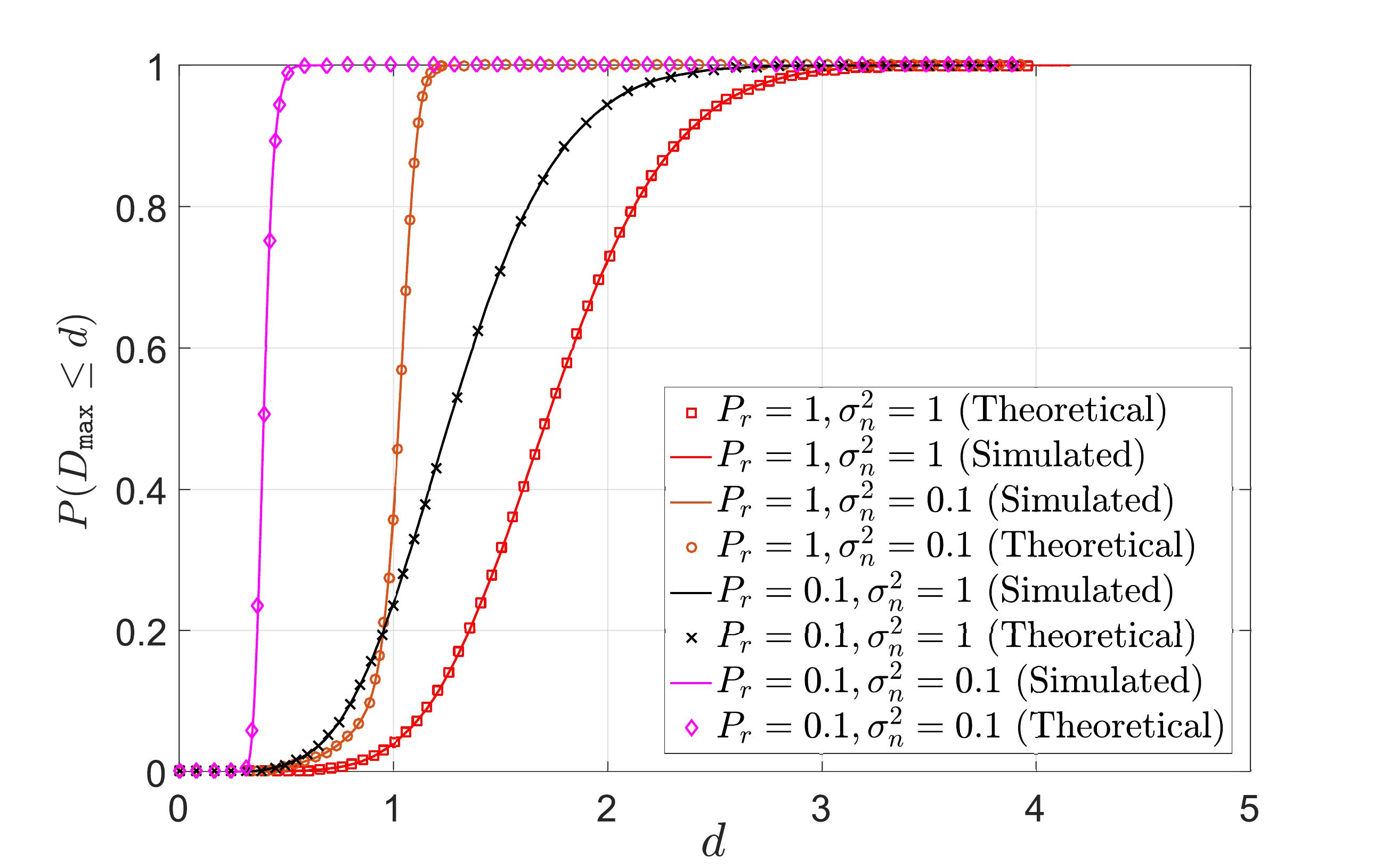}\\
		[-2ex]
		\caption{}
		\label{Fig_CDF_Heur_Theor_Sim_16QAM}	
	\end{subfigure}
	~
	\begin{subfigure}[t]{0.48\textwidth}
		\centering
		\includegraphics[width=3.2in]{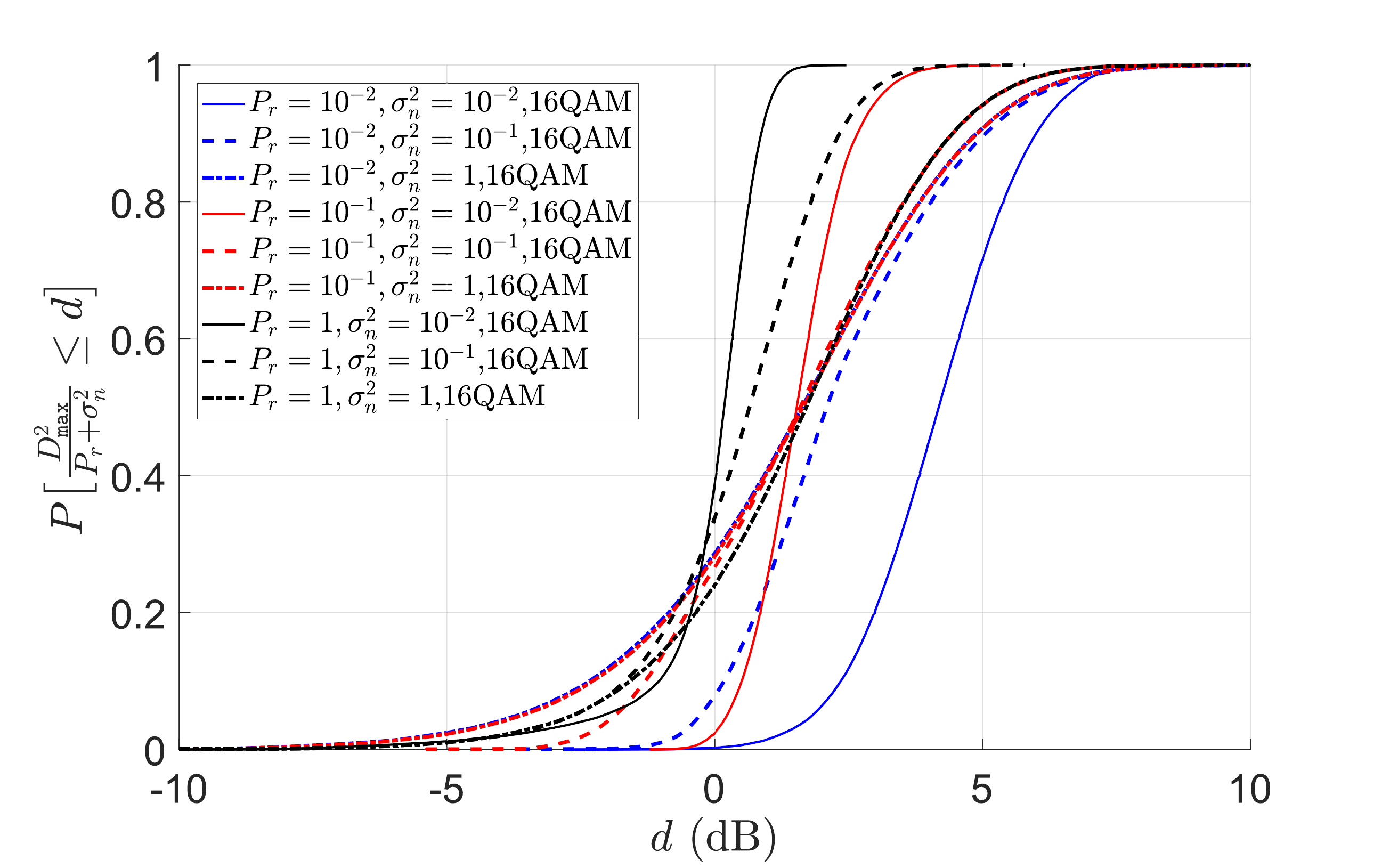}\\
		[-2ex]
		\caption{}
		\label{Fig_Heur_Ratio_Performance_16QAM}
	\end{subfigure}
	~
	\begin{subfigure}[t]{0.48\textwidth}
		\raggedleft
		\includegraphics[width=3.2in]{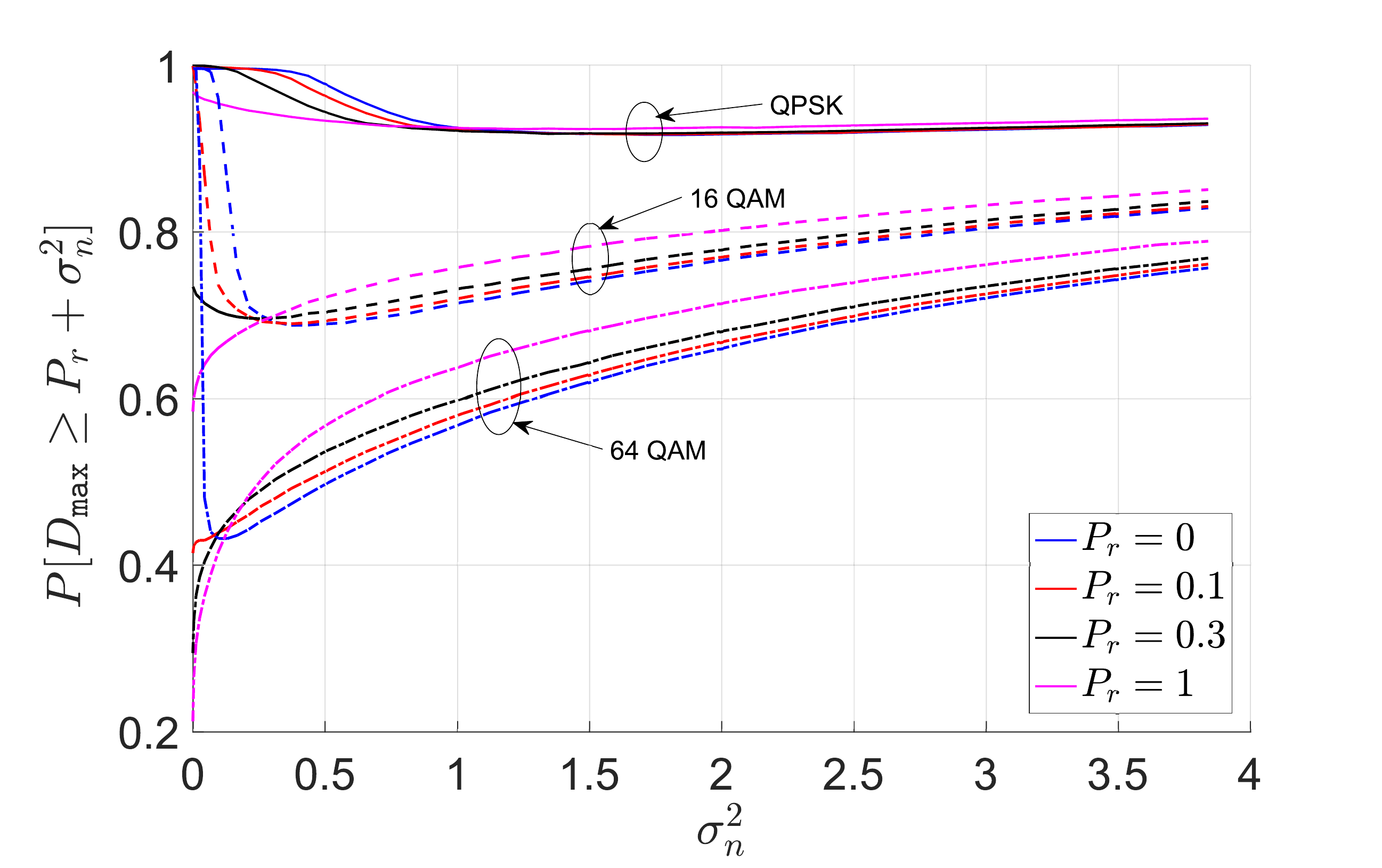}\\
		[-2ex]
		\caption{}
		\label{Fig_Robustness_all}
	\end{subfigure}
	~
	\begin{subfigure}[t]{0.48\textwidth}
		\centering
		\includegraphics[width=3.2in]{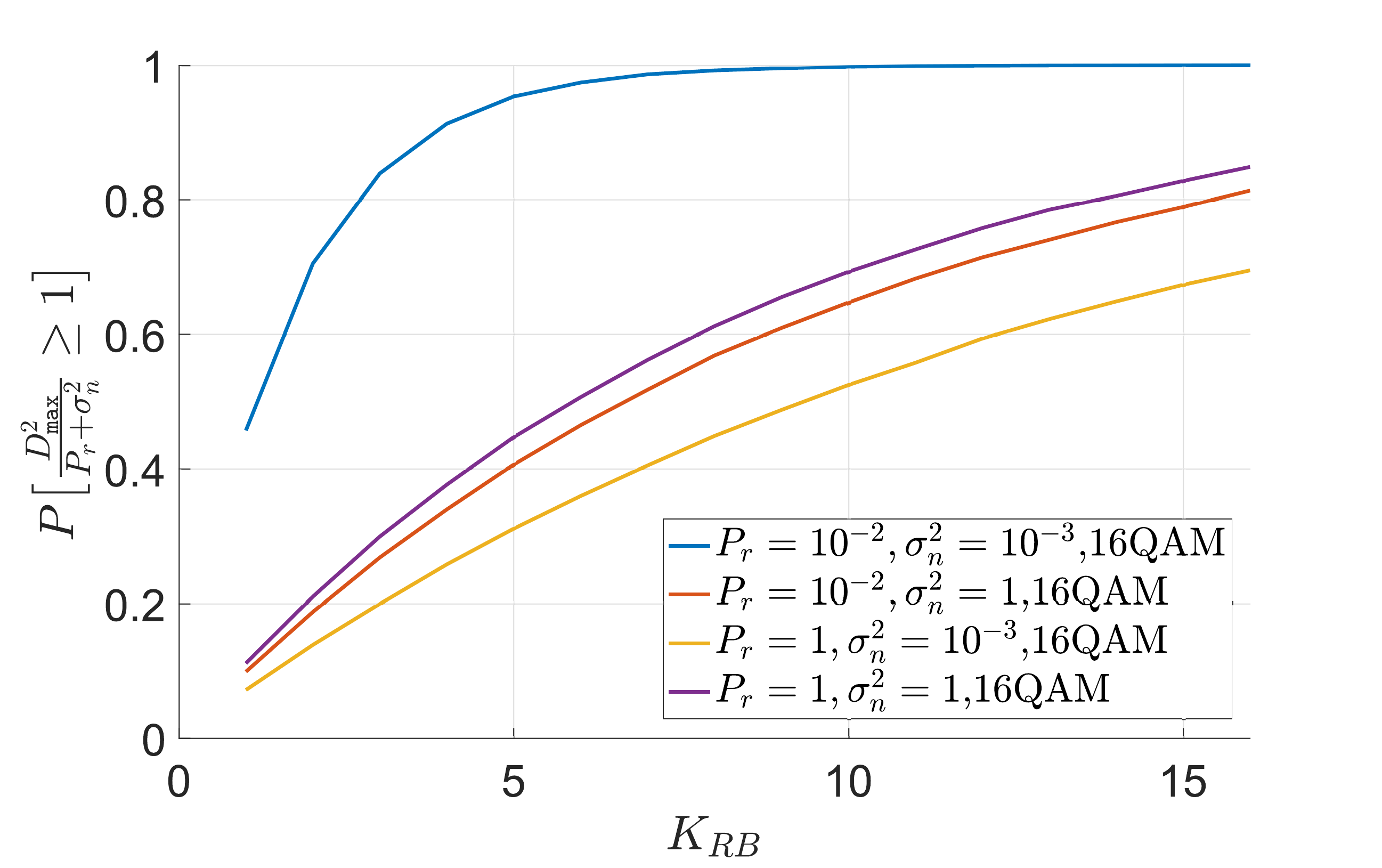}\\
		[-2ex]
		\caption{}
		\label{Fig_Robustness_16QAM_diff_M}
	\end{subfigure}
	~
	\begin{subfigure}[t]{0.47\textwidth}
		\raggedleft
		\includegraphics[width=3.2in]{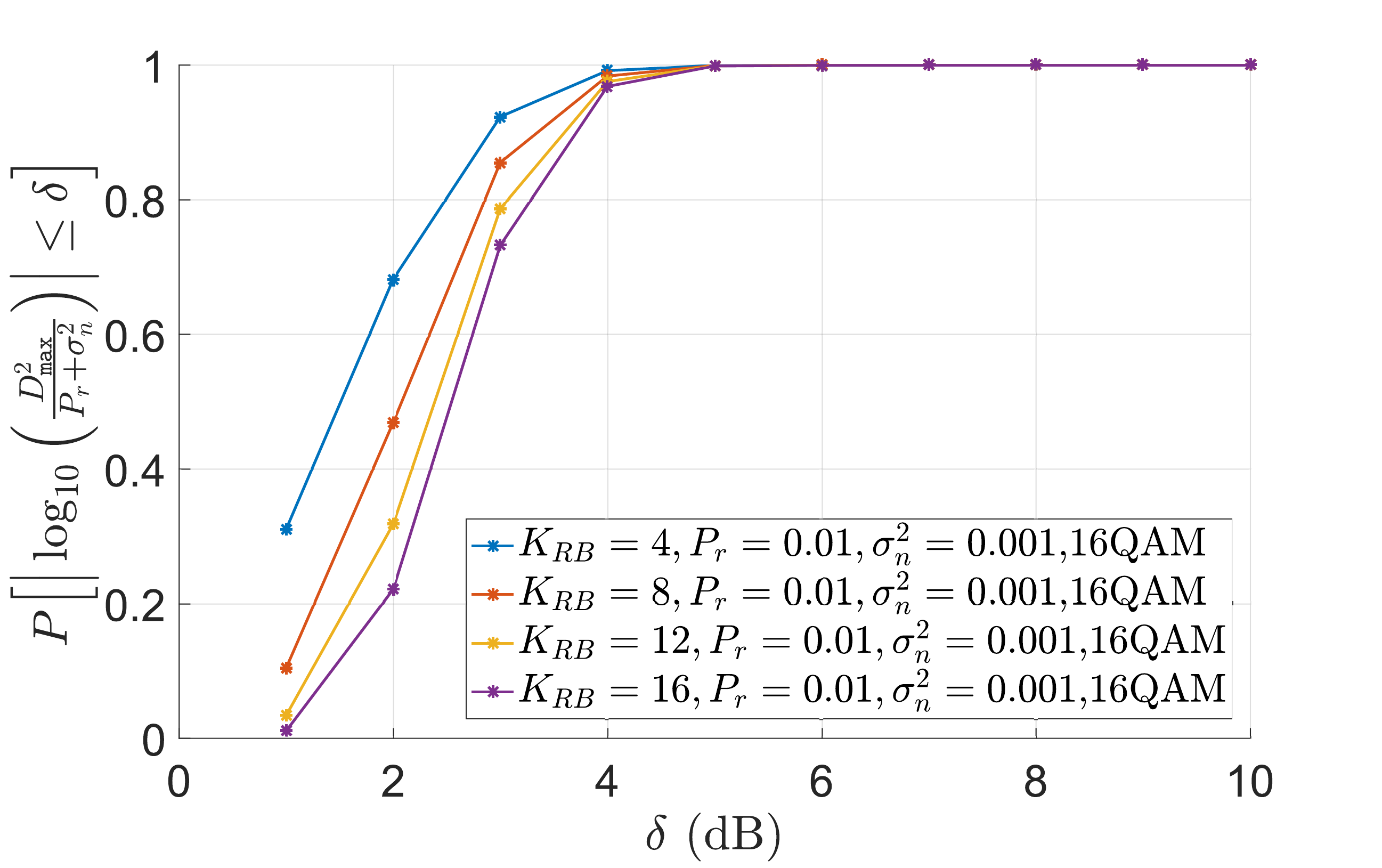}\\
		[-2ex]
		\caption{}
		\label{Fig_Accuracy_different_M_low_P_low_N}
	\end{subfigure}
	~
	\begin{subfigure}[t]{0.48\textwidth}
		\centering
		\includegraphics[width=3.2in]{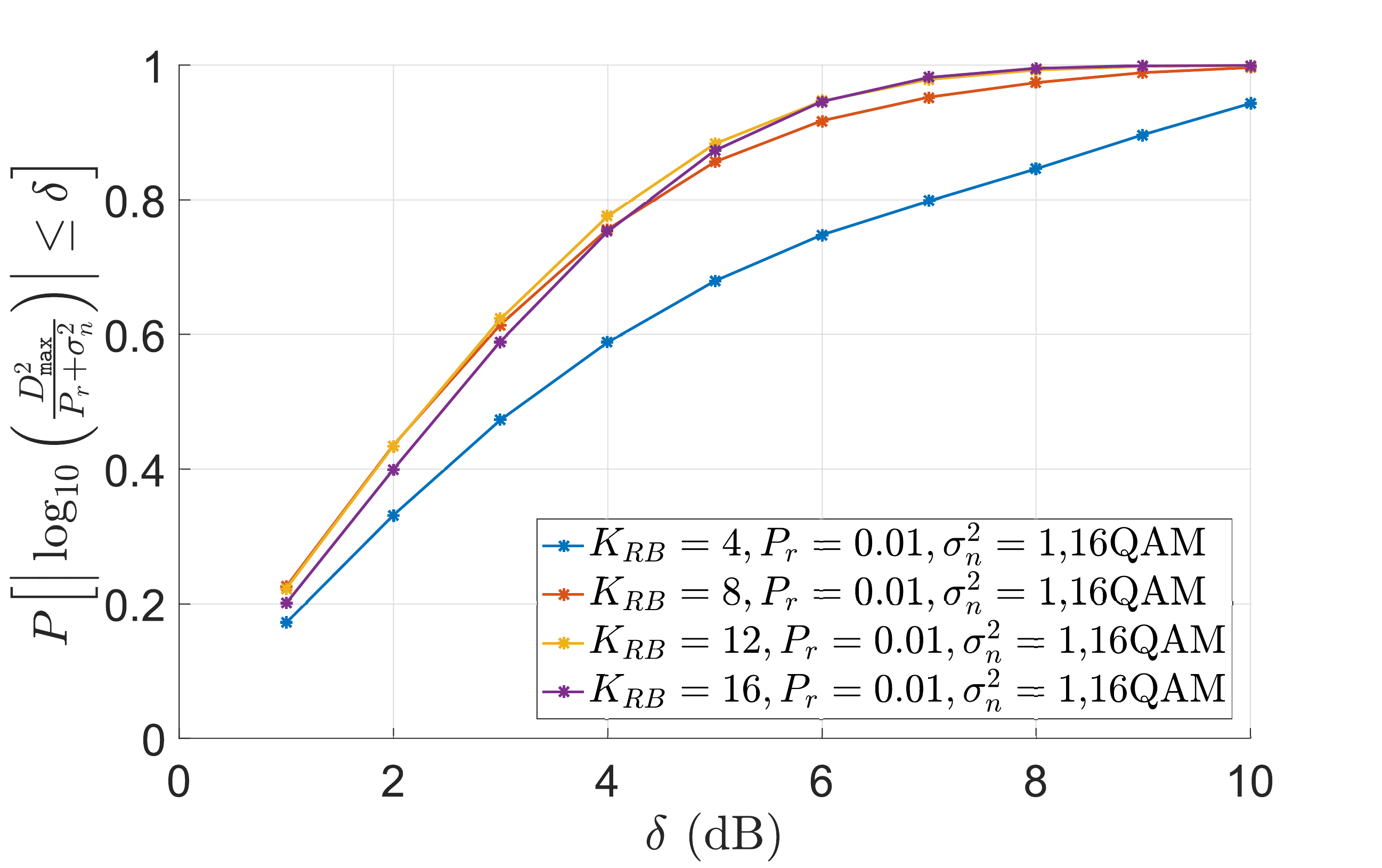}\\
		[-2ex]
		\caption{}
		\label{Fig_Accuracy_different_M_low_P_high_N}
	\end{subfigure}
	~
	\begin{subfigure}[t]{0.48\textwidth}
		\raggedleft
		\includegraphics[width=3.2in]{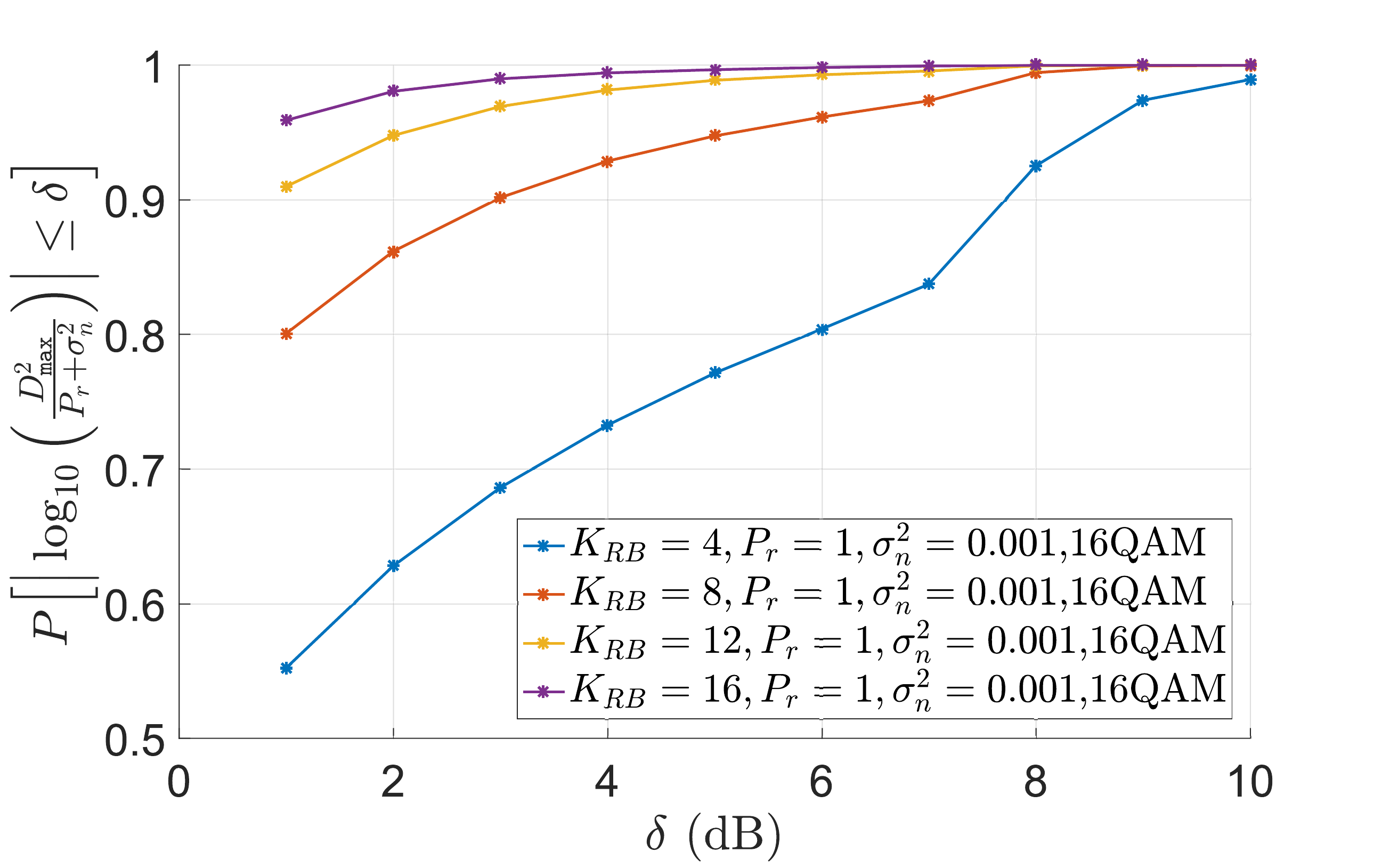}\\
		[-2ex]
		\caption{}
		\label{Fig_Accuracy_different_M_high_P_low_N}
	\end{subfigure}
	~
	\begin{subfigure}[t]{0.48\textwidth}
		\centering
		\includegraphics[width=3.2in]{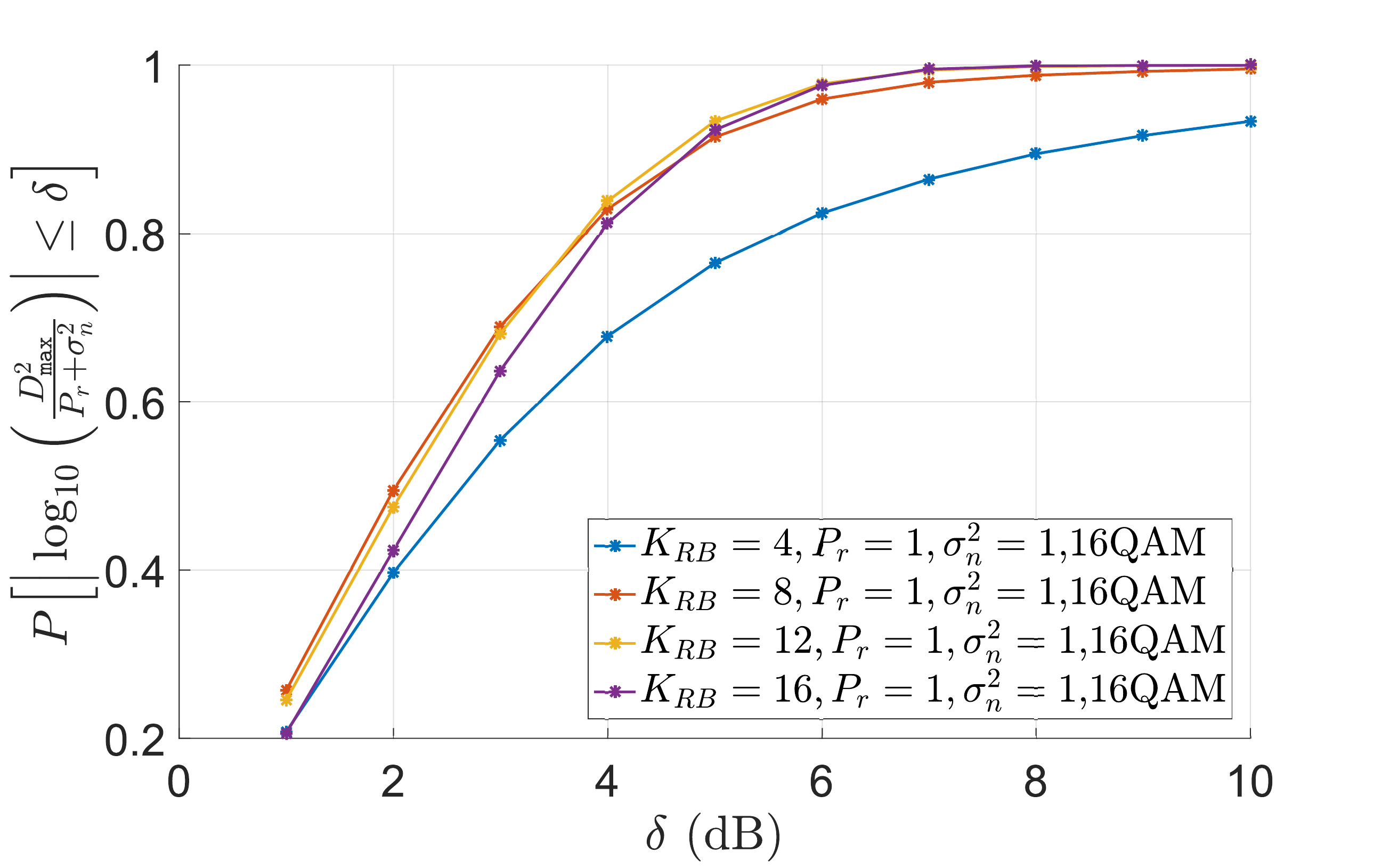}\\
		[-2ex]
		\caption{}
		\label{Fig_Accuracy_different_M_high_P_high_N}
	\end{subfigure}\\
	[-2ex]
	\caption{(a) Comparison of the simulated and theoretical distribution of $D_\mathtt{max}$ (equations (\ref{CDF_Dmin_Total_Prob_Thm}))-(\ref{Marginal_CDF_D_min}) for 16-QAM, (b) distribution of $D^2_\mathtt{max}/(P_r + \sigma^2_n)$ for 16-QAM, and (c) probability of overestimation $P_\mathtt{overest}(P_r, \sigma^2_n)$ for QPSK, 16-QAM and 64-QAM, (d) $P_\mathtt{overest}(P_r, \sigma^2_n)$ as a function of $K_{RB}$, and the accuracy metric for $K_{RB} \in \{4,8,12,16 \}$ for (e) $(P_r, \sigma^2_n) = (10^{-2}, 10^{-3})$, (f) $(P_r, \sigma^2_n) = (10^{-2}, 1)$, (g) $(P_r, \sigma^2_n) = (1, 10^{-3})$, and (h) $(P_r, \sigma^2_n) = (1, 1)$.}
	\label{Fig_Heuristic_perf_all_mods}
\end{figure}

\subsection{Accuracy and Robustness as a Function of $K_{RB}$}
Intuitively, decreasing $K_{RB}$ reduces the overestimation probability of the heuristic. This behavior can be mathematically explained as follows. Since $F_D(d) \in [0,1]$, increasing $K_{RB}$ decreases the value of the CDF $\mathbb{P}[D_\mathtt{max} \leq d]$, thereby increasing $\mathbb{P}[D_\mathtt{max} \geq d]$. Since robustness is characterized by the overestimation probability $\mathbb{P}[D_\mathtt{max} \geq \sqrt{P_r + \sigma^2_n}]$, by setting $d = \sqrt{P_r + \sigma^2_n}$, we see that decreasing $K_{RB}$ decreases $\mathbb{P}[D_\mathtt{max} \geq \sqrt{P_r + \sigma^2_n}]$ and vice-versa. Fig. \ref{Fig_Robustness_16QAM_diff_M} shows the robustness (probability of overestimating interference-plus-noise-power) of the heuristic for 16-QAM  as a function of $K_{RB}$, for different values of $(P_r, \sigma^2_n)$.

In contrast, it is not straightforward to infer the dependence of accuracy on $K_{RB}$ using mathematical arguments and hence, we use numerical studies to do the same. Figures \ref{Fig_Accuracy_different_M_low_P_low_N}-\ref{Fig_Accuracy_different_M_high_P_high_N} show the accuracy as a function of $\delta$ (dB) for $K_{RB} \in \{4,8,12,16\}$, for different values of ($P_r,\sigma^2_n$). We observe that (a) accuracy is not a monotonic function of $K_{RB}$, and (b) the optimal $K_{RB}$ that maximizes the accuracy of the SINR estimate depends on $(P_r, \sigma^2_n)$, as well as the accuracy threshold $\delta$. In addition, we notice that (a) in high SINR regimes, a low $K_{RB}$ ensures high accuracy (Fig. \ref{Fig_Accuracy_different_M_low_P_low_N}), (b) in interference-limited scenarios (high $P_r$ and low $\sigma^2_n$), a high $K_{RB}$ value ensures high accuracy (Fig. \ref{Fig_Accuracy_different_M_high_P_low_N}), and (c) in intermediate noise and interference conditions, a $K_{RB}$ value of $8-16$ yields similar accuracy performance (Figs. \ref{Fig_Accuracy_different_M_low_P_high_N} and \ref{Fig_Accuracy_different_M_high_P_high_N}).

Similar trends are observed for QPSK and 64-QAM. Unfortunately, a comprehensive mathematical analysis of the accuracy is beyond the scope of this paper. The key takeaway from Figures \ref{Fig_Accuracy_different_M_low_P_low_N}-\ref{Fig_Accuracy_different_M_high_P_high_N} is that there is no universal $K_{RB}$ value that maximizes the accuracy of the heuristic-aided SINR estimate. However, memory-based schemes that leverage knowledge of interference and noise conditions in the recent past, can be used to choose $K_{RB}$ to balance the robustness and accuracy of the heuristic-aided SINR estimate.

\begin{remark}
	It is worthwhile to observe that the max-min heuristic is independent on the multicarrier waveform used, and yields accurate SINR estimates of a coherence block in the received signal. 
\end{remark}

\section{Semi-Blind/Hybrid Post-Equalizer SINR Estimation Framework} \label{Hyb_SINR_Estim_Framwrk}
In this section, we describe the \textit{`semi-blind/hybrid'} post-equalizer SINR estimation framework, which uses pilot-aided (Section \ref{Pil-Aided-SINR-Estim-Methd}) as well as heuristic-aided (Section \ref{Low_Complex_Heur_SINR_Estimate}) SINR estimates. 

Let the data block contain $N_\mathtt{blk}$ OFDM symbols\footnote{In LTE and NR, the data block is termed as the transport block, which is often sent over a subframe consisting of 14 OFDM symbols.}, $\mathcal{A}_{NP}$ be the set of non-pilot OFDM symbol indices, and $k \in \mathcal{K}[m]$ be the subcarrier indices of data resource elements in the $m^{th}$ OFDM symbol. The SINR of the RE on the $k^{th}$ subcarrier of the $n^{th}$ OFDM symbol can be estimated using
\begin{align}
\label{SINR_Estim_Per_RE}
\hat{\gamma}[n,k] = \begin{cases}
\hat{\gamma}_p[n,k] & \text{for } n=1,2,\cdots, N_\mathtt{blk}, \mathbbm{1}[n], \text{ and } n\notin \mathcal{A}_{NP} \\
\hat{\gamma}_p [n,k] & \text{if } n\neq m, \mathbbm{1}[m], \text{ and } m \in \mathcal{A_{NP}} \\
\frac{1}{D_\mathtt{max}[m,k]} & \text{if } \mathbbm{1}[n] \text{ and } n \in \mathcal{A}_{NP},
\end{cases}
\end{align}
where $\mathbbm{1}[m]$ denotes the occurrence of pulsed radar interference on the $m^{th}$ OFDM symbol, $D_\mathtt{max} [m,k]$ is the heuristic for every RE in the coherence block of the contaminated OFDM symbol. If the coherence block contains $K_{RB}$ subcarriers, then  $D_\mathtt{max}[m,lK_{RB}+1] = \cdots = D_\mathtt{max}[m,(l+1)K_{RB}]$ for $l \in \mathbb{Z}$. 

To determine the appropriate SINR estimate to be used, the contaminated OFDM symbol needs to be known. As shown in Fig. \ref{SINR_Estim_Framework}, the following intermediate stages are necessary to acquire this information in practice:
\begin{enumerate}
\item Pulsed radar parameter estimation,
\item Detection of the pilot symbol interference, and
\item Detection of contaminated OFDM symbol index.
\end{enumerate}
\vspace{-5pt}
\subsection{Pulsed Radar Parameter Estimation}\label{subsec:T_rep_estim}
Most pulsed radars have a fixed repetition interval ($T_\mathtt{rep}$) for an extended duration of time (timescale of seconds). Since the interference is periodic, $T_\mathtt{rep}$ can be estimated by applying Fourier techniques on time-series data of received power per data block, resulting in a low-complexity baseband implementation. Subsequently, the UE can predict future subframes indices which will be impaired by radar interference. 
\vspace{-5pt}
\subsection{Threshold-based Detection of Pilot Interference}\label{Pilot_Interf_Detect_Threshold}
Interference on pilot symbols result in accurate SINR estimates \cite{Rao_VTC_LTE_LinkAdapt_2019}. Pilot interference can be detected by monitoring pilot-aided SINR estimates in every data block. For the $k^{th}$ data block, the receiver calculates the wideband SINR metric $\hat{\gamma}_{\mathtt{avg},p}[k]$ using pilot-aided methods. Using knowledge of $T_\mathtt{rep}$ and pilot-aided SINR estimates of previous data blocks, the wideband SINR metric for non-pilot radar interference ($\hat{\gamma}_{\mathtt{NPI},p}$) is computed.
If the current ($k^{th}$) block is impaired by interference, then 
\begin{enumerate}
\item if $\hat{\gamma}_{\mathtt{NPI},p} - \hat{\gamma}_{\mathtt{avg},p}[k] \geq \gamma_\mathtt{th}$, the $k^{th}$ block is considered to be impaired by pilot interference, and
\item if $\hat{\gamma}_{\mathtt{NPI},p} - \hat{\gamma}_{\mathtt{avg},p}[k] < \gamma_\mathtt{th}$, the $k^{th}$ block is considered to be impaired by non-pilot interference. 
\end{enumerate}
In practice, a typical value of the threshold is $\gamma_\mathtt{th} = 1$ dB, since the channel quality indicator (CQI) remains the same with a high probability for a SINR mismatch of $\pm 1$ dB \cite{rupp2016vienna}. 
\vspace{-5pt}
\subsection{Log Likelihood-based Detection of the Interference-Impaired OFDM Symbol}\label{subsec:OFDM_index_corrupted}
OFDM has a long symbol duration ($72\ \mu s$ in sub-6 GHz bands of LTE and NR). Hence, for wideband radars with a short pulse width ($T_\mathtt{pul}\sim 1\ \mu$s), the probability of two adjacent OFDM symbols being contaminated is \textit{almost zero}\footnote{In sub-6 GHz systems, the typical cyclic prefix duration is $5-10\ \mu$s. A radar pulse time-aligned with \textit{two consecutive OFDM symbols} will lie within the cyclic prefix (CP) of the second symbol. Due to CP removal in OFDM, radar interference will not impact the second OFDM symbol in such scenarios.} for sub-6 GHz cellular systems. Therefore, we ignore the possibility of multiple adjacent OFDM symbols being interfered. 

We use a log likelihood-based approach to detect the contaminated data-bearing OFDM symbol in every block, which is executed when pilots are detected to be interference-free. Algorithm \ref{Algo_detect_contamin_OFDM_symb} shows the proposed approach if the estimated $T_\mathtt{rep}$ indicates a single radar pulse within the data block\footnote{If $T_\mathtt{rep}$ estimates indicate that $m$ radar pulses will impair the data block, then Algorithm \ref{Algo_detect_contamin_OFDM_symb} outputs indices corresponding to the $m$ least values.}. The empirical log likelihood function models the hypothesis of noise-only impairment ($P_r = 0$), and is calculated for each non-pilot OFDM symbol. Intuitively, interference-free symbols statistically have a smaller nearest neighbor distance when compared to impaired symbols. In a coherence block, pilot-aided SINR estimates are constant for all REs. As a result, the proposed approach has a high probability of accurately detecting the impaired OFDM symbol.

\begin{algorithm}[t]
	\small
	\begin{algorithmic}[1]
		\State \textbf{Input:} In each data block, 
		\Statex \textit{Set of non-pilot OFDM symbol indices $\mathcal{A}_{NP}$} 
		\Statex \textit{Data subcarriers of $n^{th}$ OFDM symbol $\mathcal{K}[n]$}
		\Statex \textit{Post-processed OFDM symbols $y[n,k]\ \forall\ n \in \mathcal{A}_{NP}, k \in \mathcal{K}[n]$} 
		\State Find nearest neighbor of each $y[n,k]$ using $\hat{x}_\mathtt{nn}[n,k] = \underset{x \in \mathcal{X}}{\text{arg min }} \| y[n,k] - x\|_2\ \forall\ n \in \mathcal{A}_{NP}, k \in \mathcal{K}[n]$. 
		\State For each $(n,k)$, obtain the pilot-aided SINR $\hat{\gamma}_p[n,k]$ using (\ref{SINR_pilot}). 
		\State The contaminated OFDM symbol index ($\hat{n}$) is detected by minimizing the log-likelihood function using
		\begin{align}
		\label{Lklhood_hrstic_contam_OFDM_symb}
		\hat{n} = \underset{n \in \mathcal{A}_{NP}}{\text{arg min }} \frac{-1}{|\mathcal{K}[n]|} \sum_{k \in \mathcal{K}[n]} \hat{\gamma}_p[n,k] |y[n,k] - \hat{x}_\mathtt{nn}[n,k] |^2.
		\end{align}
		\State Go back to step 1 in the next data block.
	\end{algorithmic}
	\caption{Detection of Corrupted OFDM Symbol Index}
	\label{Algo_detect_contamin_OFDM_symb}
\end{algorithm}

\subsection{SINR Estimation Using Data Block Reconstruction}
If the transmitted symbols are known, then the post-equalizer SINR can be estimated at the receiver perfectly. If $x[n,k] \in \mathcal{X}$ is the transmitted symbol on RE $(n,k) \in \mathcal{D}$, and $y[n,k]$ is the corresponding post-processed received symbol. The post-processing SINR of RE $(n,k)$ can be directly estimated using 
\begin{align}
\gamma[n,k] = \frac{|x[n,k]|^2}{|x[n,k] - y[n,k]|^2}, 
\end{align}
and the correspond wideband SINR metrics (average or EESM-based) can be estimated using (\ref{SINR_Mapping_EESM_avg}). But $x[n,k]$ can seldom be accurately estimated in the presence of noise interference. However, it can be perfectly reconstructed if the post-decoder bit sequence is known to be accurate. 

If $\mathbf{b}$ represents the data bits after turbo-decoding, the receiver can reconstruct the transmitted data symbol on each RE by implementing the transmitter baseband processing chain\footnote{Since 3GPP standardization documents are publicly available \cite{3GPPLTE_TS36213_v12}, it is possible for the receiver to implement the transmitter processing chain if the appropriate control information is decoded correctly.}. However, perfect reconstruction is guaranteed only when $\mathbf{b}$ is accurate. In LTE and NR, the integrity of $\mathbf{b}$ is ensured using a cyclic redundancy check (CRC) at the end of each data block, where $\text{CRC}=0\ (1)$ indicates decoding success (failure). Since an $n$-bit CRC has a false positive rate of $2^{-n}$, (where $n=24$ in LTE and NR \cite{dahlman20185g}), we use the CRC as an indicator to accurately reconstruct $x[n,k]$ in our proposed framework. 
\begin{remark}
If $x_\mathtt{nn}[n,k]$ is the nearest neighbor of $y[n,k]$, then $|x[n,k] - y[n,k]|^2 \geq |x_\mathtt{nn}[n,k] - y[n,k]|^2$. For constant envelope modulation schemes, post-equalizer SINR estimated using the nearest neighbor decision rule forms an upper bound to the actual SINR. In other QAM schemes, nearest neighbor association often overestimates the SINR. 
\end{remark}

\begin{figure}[!htbp]
	\centering
	\begin{subfigure}[t]{0.47\textwidth}
		\centering
		\includegraphics[width=2.8in]{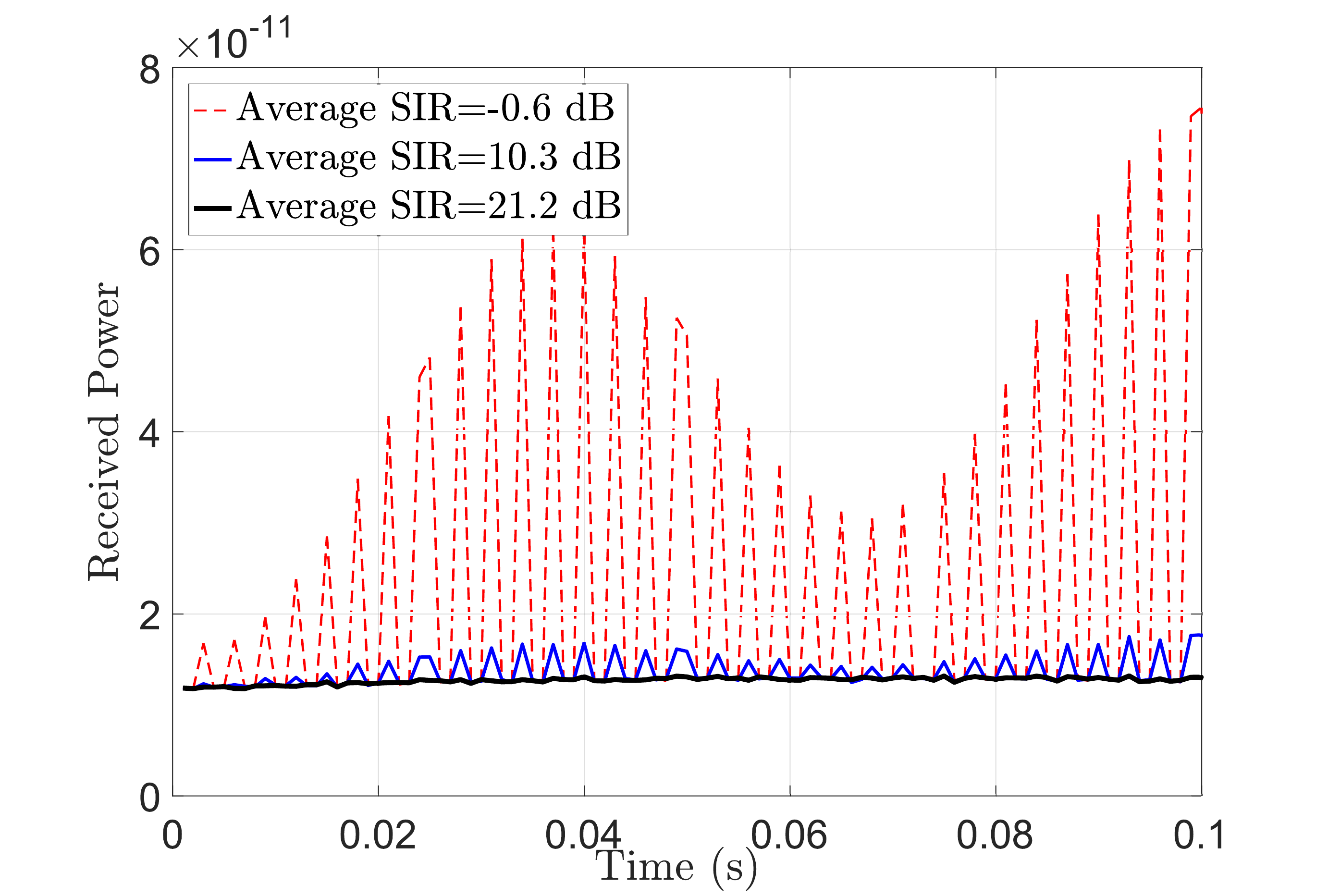}\\
		[-2ex]
		\caption{}
		\label{Fig_Received_Power}
	\end{subfigure}
	~ 
	\begin{subfigure}[t]{0.47\textwidth}
		\centering
		\includegraphics[width=2.8in]{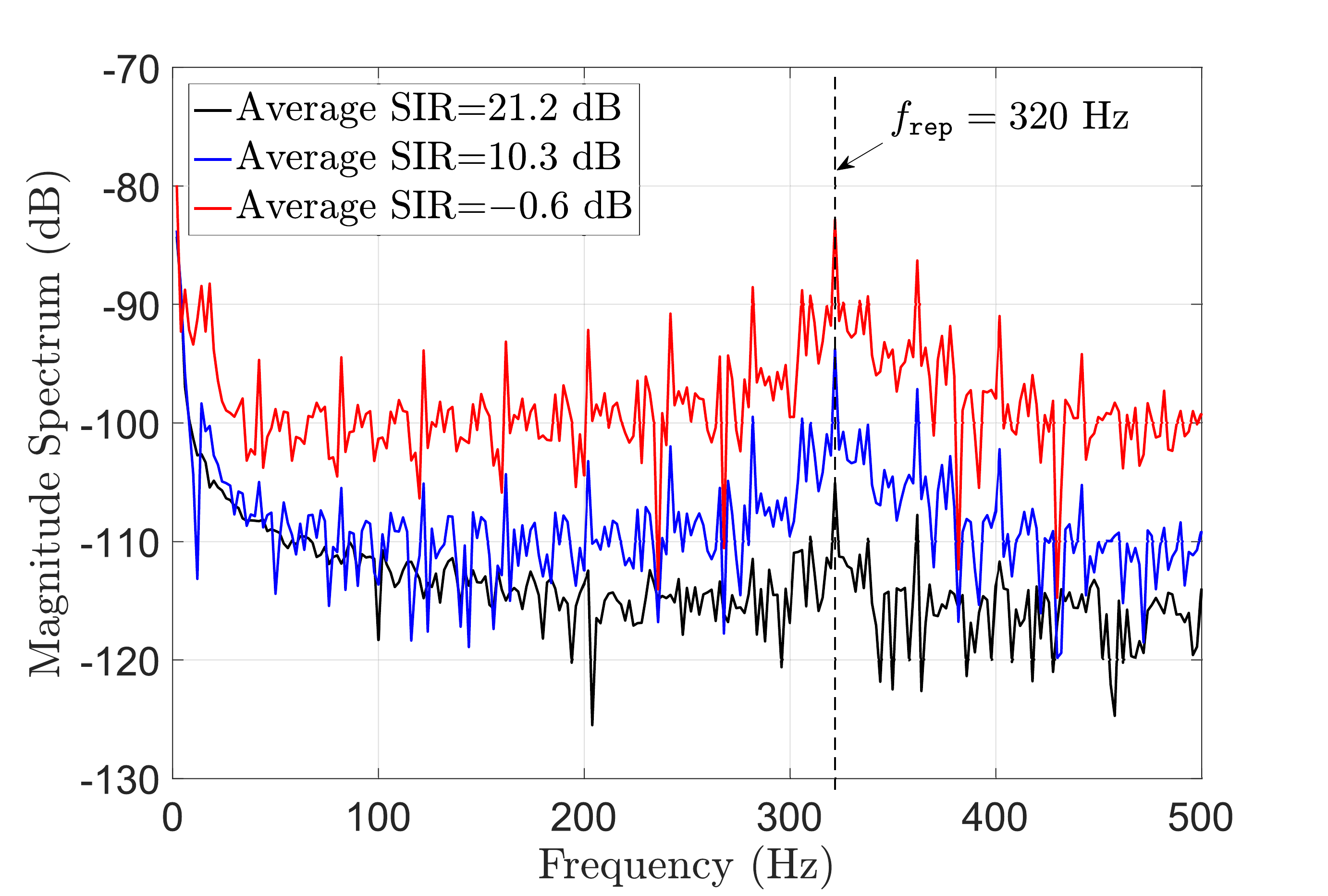}\\
		[-2ex]
		\caption{}
		\label{Fig_FFT_Received_Power}
	\end{subfigure}
	~
	\begin{subfigure}[t]{0.47\textwidth}
		\centering
		\includegraphics[width=2.8in]{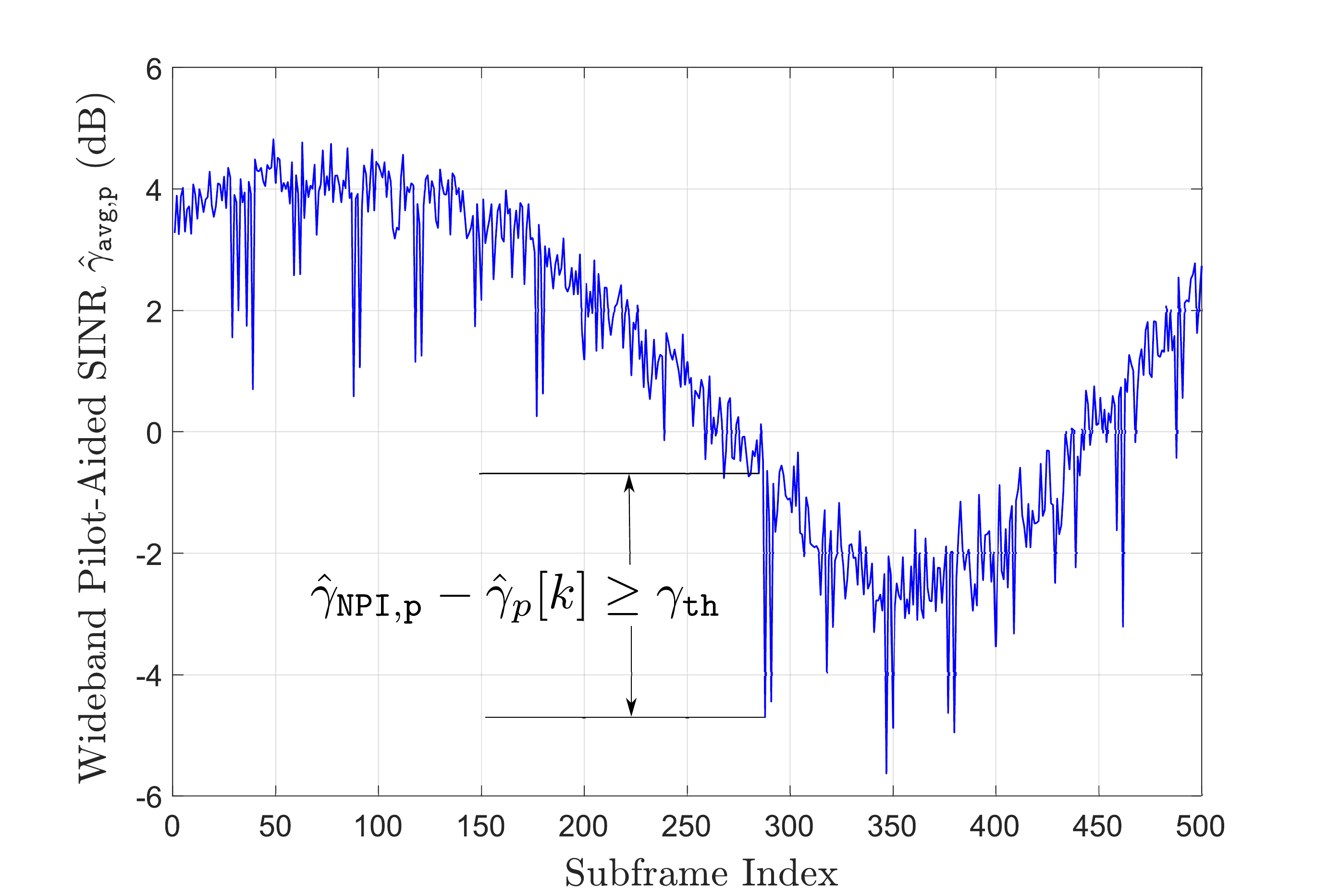}\\
		[-2ex]
		\caption{}
		\label{Fig_Threshold_based_pil_contam_det}
	\end{subfigure}
	~ 
	\begin{subfigure}[t]{0.47\textwidth}
		\centering
		\includegraphics[width=2.8in]{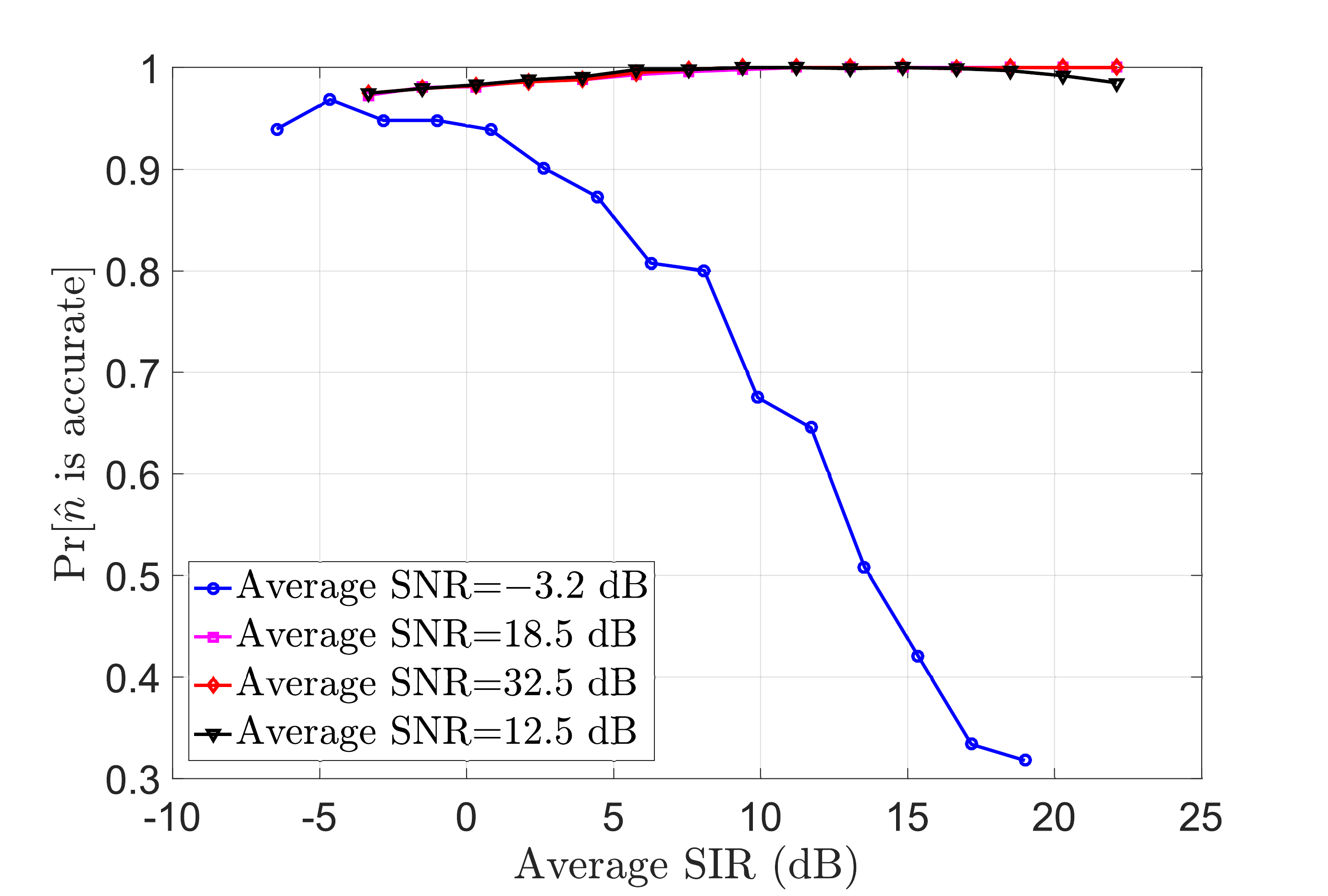}\\
		[-2ex]
		\caption{}
		\label{Fig_Prob_Contam_OFDM_symb_Correct}
	\end{subfigure}
	~ 
	\begin{subfigure}[t]{0.47\textwidth}
		\centering
		\includegraphics[width=2.8in]{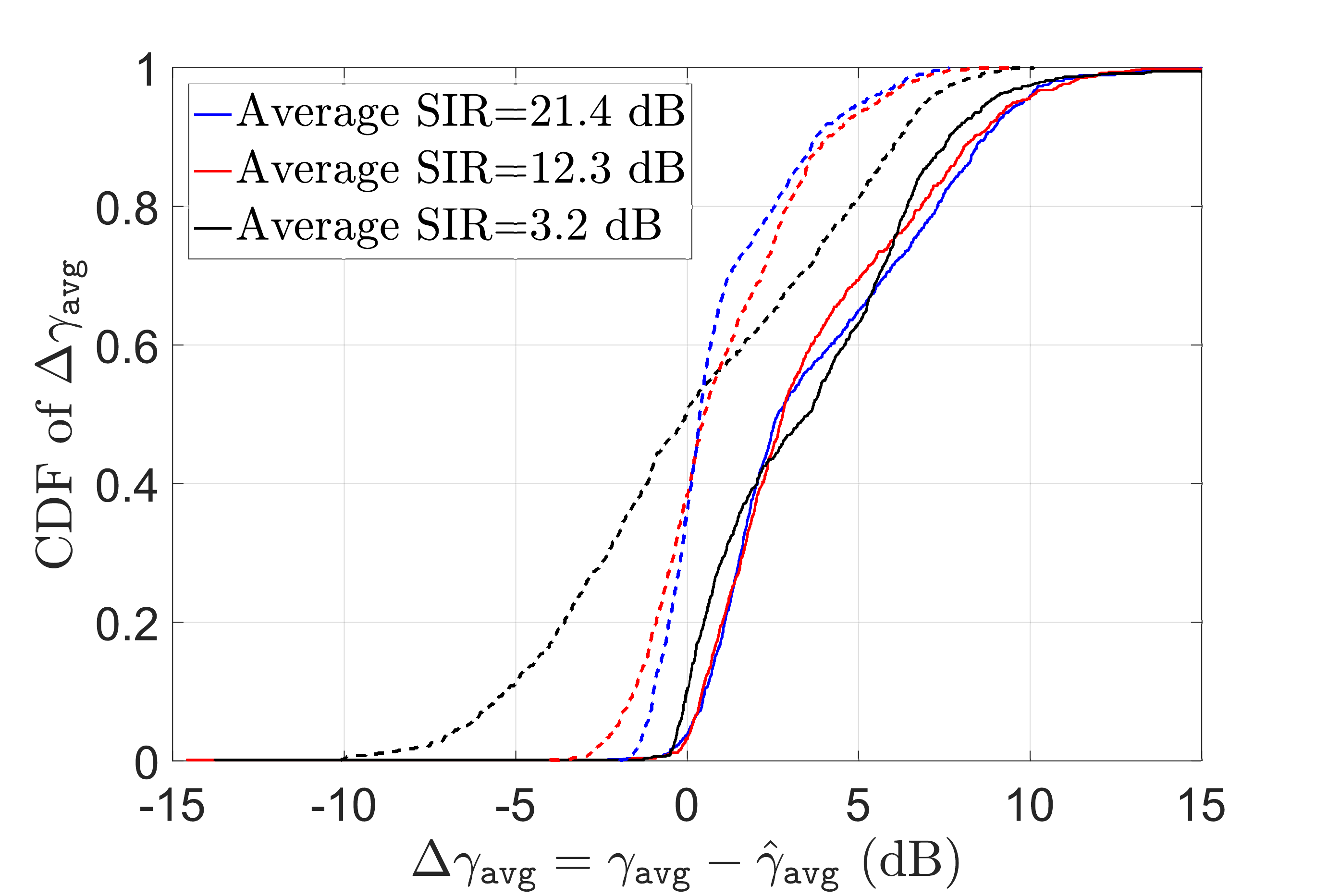}\\
		[-2ex]
		\caption{}
		\label{CDF_Mismatch_SINR_very_low_SNR}
	\end{subfigure}
	~ 
	\begin{subfigure}[t]{0.47\textwidth}
		\centering
		\includegraphics[width=2.8in]{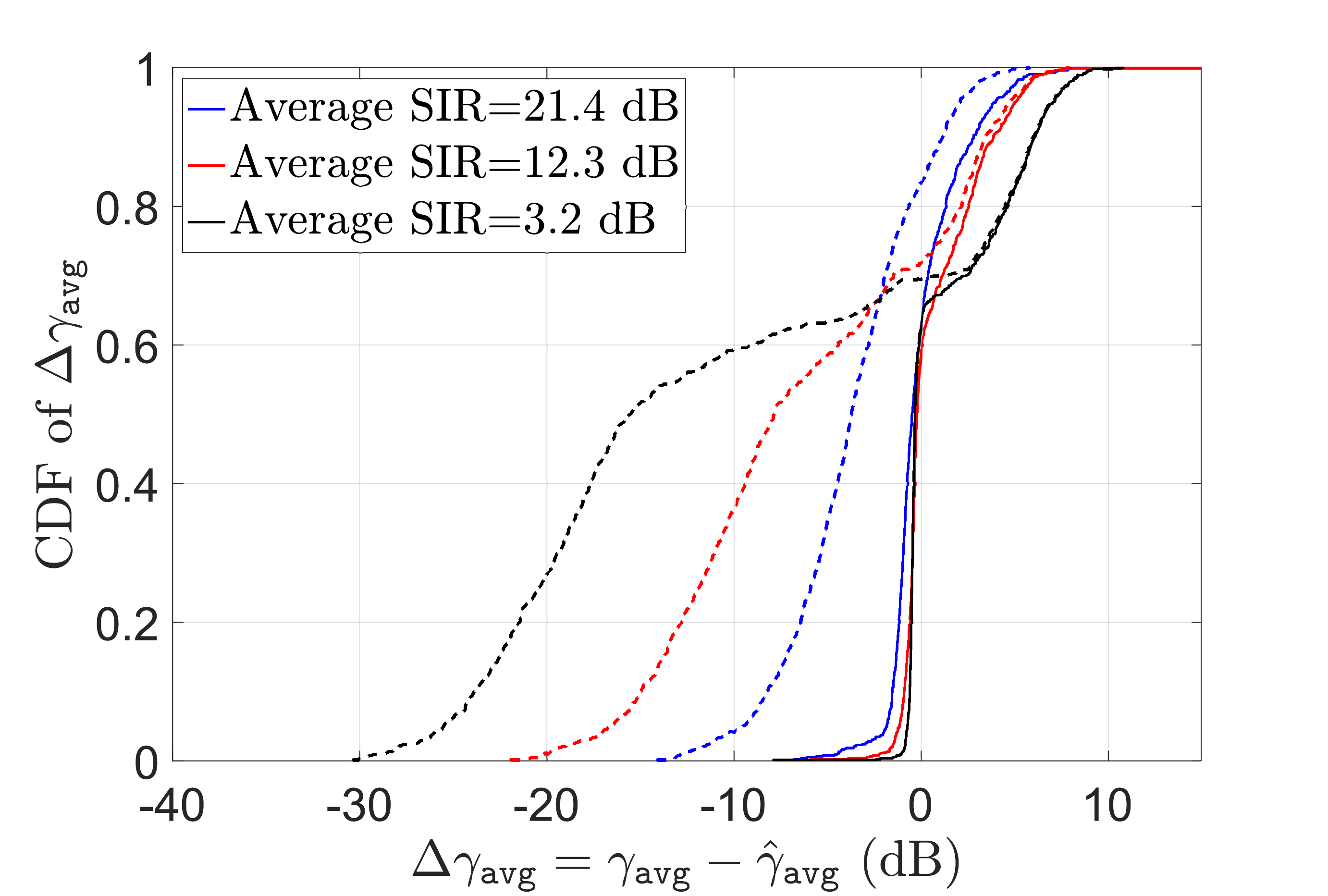}\\
		[-2ex]
		\caption{}
		\label{CDF_Mismatch_SINR_low_SNR}
	\end{subfigure}
	~ 
	\begin{subfigure}[t]{0.47\textwidth}
		\centering
		\includegraphics[width=2.8in]{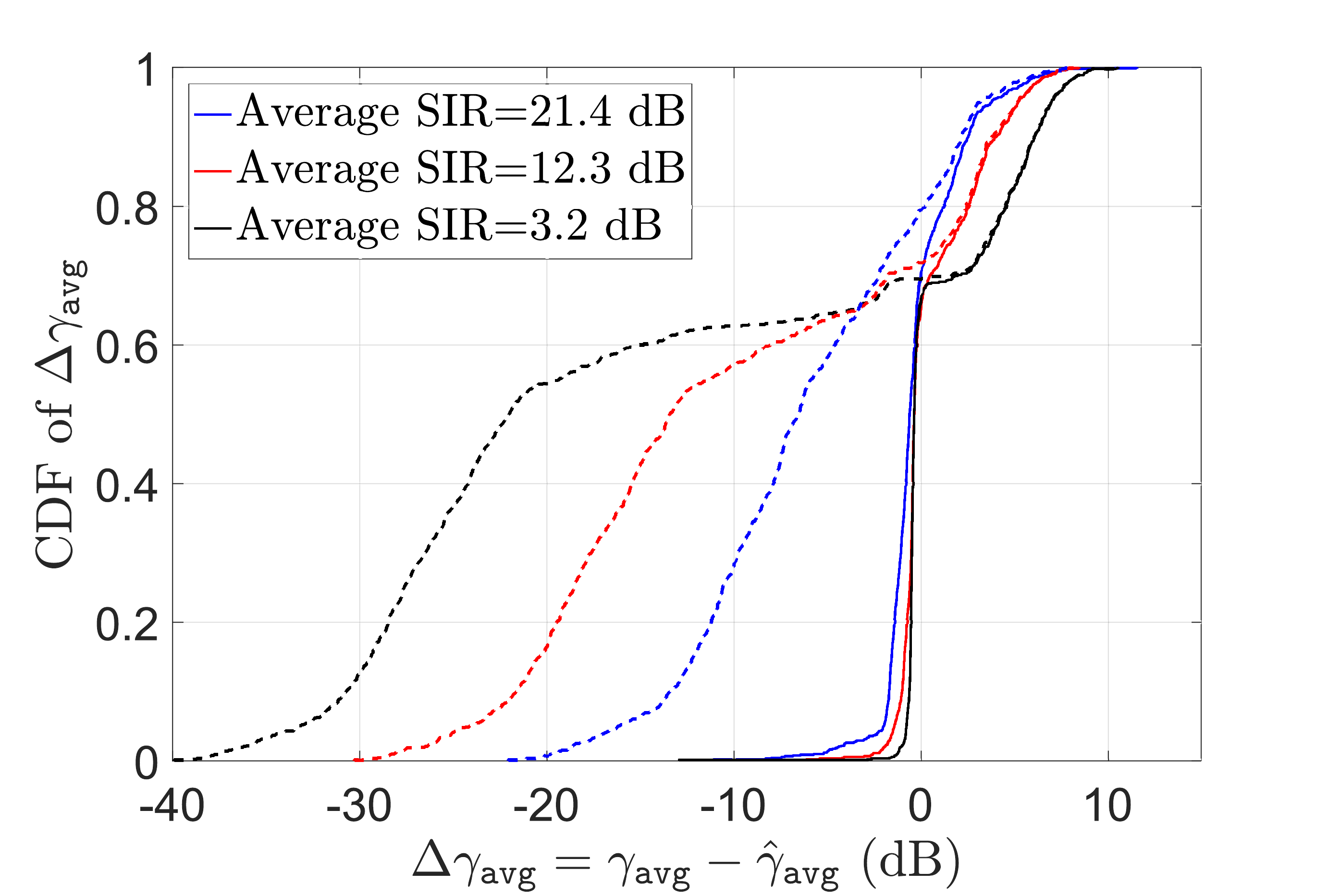}\\
		[-2ex]
		\caption{}
		\label{CDF_Mismatch_SINR_medium_SNR}
	\end{subfigure}\\
	[-2ex]
	\caption{Performance of various stages of the hybrid post-equalizer SINR estimation framework: (a) received power per subframe in the LTE downlink with average $\text{SNR}=19.5$ dB, (b) corresponding amplitude spectrum of the received power per subframe for a window length of $500$ points, (c) illustration of the threshold-based pilot contamination detection with $\gamma_\mathtt{th} = 1$ dB when average $\text{SNR}=-0.2$ dB, (d) probability of accurate \textit{contaminated symbol} detection, and comparison of the average SINR metric mismatch ($\Delta \gamma_\mathtt{avg} \text{ (dB) } = \gamma_\mathtt{avg} \text{ (dB) } - \hat{\gamma}_\mathtt{avg} \text{ (dB)} $) for interference-impaired subframes using the proposed framework (in solid lines) and the pilot-aided method (in dashed lines) using equation (\ref{SINR_pilot}), when (e) average $\text{SNR}=-0.2$ dB, (f) average $\text{SNR}=13.8$ dB, and (g) average $\text{SNR}=19.5$ dB.}
	\label{Fig_Num_Results_SINR_Estim_Framework}
\end{figure}
\subsection{Numerical Results} \label{Num_Results_SINR_Estim_Framework}
In this subsection, we show the performance results of the proposed SINR estimation framework. We consider the example of the LTE-A Pro downlink sharing spectrum with a linear frequency modulated pulsed radar with the transmitted waveform shown in (\ref{Radar_waveforms}), and the other system parameters shown in Table \ref{Table_Radar_LTE_coexist_HybSINR_DCFB}. In addition, the assumptions used to analytically characterize the heuristic performance in section \ref{Low_Complex_Heur_SINR_Estimate} are relaxed in the numerical results presented below.

Fig. \ref{Fig_Received_Power} shows the downlink received power in every subframe. The corresponding windowed FFT computed using a window length of $500$ subframes is shown in Fig. \ref{Fig_FFT_Received_Power}. We observe that the amplitude spectrum can accurately estimate $f_\mathtt{rep}=\tfrac{1}{T_\mathtt{rep}}$ for a wide range of SIR values. 

Fig. \ref{Fig_Threshold_based_pil_contam_det} illustrates the threshold-based pilot contamination detection method described in section \ref{Pilot_Interf_Detect_Threshold}. Using $f_\mathtt{rep}$ and memory of received power per subframe in the recent past, the pilot-aided wideband SINR ($\hat{\gamma}_{\mathtt{NPI},p}$) is calculated for \textit{interference-free subframes}\footnote{If $T_\mathtt{rep}$ is smaller than the subframe duration, then the received power of \textit{each OFDM symbol} needs to be used to estimate $f_\mathtt{rep}$, and detect pilot contamination.}, and compared to pilot-aided SINR of the current subframe. As mentioned earlier, a threshold of $\gamma_\mathtt{th} = 1$ dB is chosen, since variations greater $\pm 2$ dB will result in use of a different MCS \cite{rupp2016vienna}.

Fig. \ref{Fig_Prob_Contam_OFDM_symb_Correct} shows the performance of Algorithm \ref{Algo_detect_contamin_OFDM_symb}, for different values of SINR. At low SNR, we observe that the accuracy of the proposed method improves with increasing INR when the interference power rises above the noise floor. For medium to high SNRs, the probability of accurate detection is greater than $95\%$, indicating reliable detection performance for a wide range of SIR and INR values. 

Fig. \ref{CDF_Mismatch_SINR_very_low_SNR}-\ref{CDF_Mismatch_SINR_medium_SNR} compares SINR estimation performance of the proposed framework ($\hat{\gamma}_\mathtt{avg,hyb}$) with the pilot-aided method ($\hat{\gamma}_\mathtt{avg,p}$) \textit{for interference-impaired subframes}. The distribution of the average SINR mismatch $\Delta \gamma_\mathtt{avg} = (\gamma_\mathtt{avg} - \hat{\gamma}_\mathtt{avg})$, for a wide range of SNR and INR conditions are plotted, where negative $\Delta \gamma_\mathtt{avg}$ indicate overestimated SINR values. We observe that pilot-aided methods have a high density of negative $\Delta \gamma_\mathtt{avg}$, that results in degradation of link adaptation performance. In contrast, the proposed framework improves the SINR estimation performance for a large range of SNR and INR values. In the low SNR-high SIR regime, we observe that the proposed framework underestimates the SINR with a probability higher than $95 \%$. This trend can be attributed to the robustness of the heuristic in QPSK, which is typically used in low SINR conditions. In other SNR and SIR regimes, we observe that the semi-blind wideband SINR estimate (a) lies within $\pm 5$ dB of the true value for more than $80\%$ of the subframes, and (b) is skewed towards conservative SINR ($\Delta \gamma_\mathtt{avg}$) estimates. As we will demonstrate in the next section, \textit{robust SINR estimates} obtained using the proposed framework significantly improves link-level performance in hostile spectrum sharing environments. However, these improvements are dependent on the availability of accurate SINR estimates for both interference-impaired and interference-free subframes. An explicit scheme to ensure the availability of accurate CSI is presented in the next section. 

\begin{figure}[!t]
	\centering
	\includegraphics[width=4in]{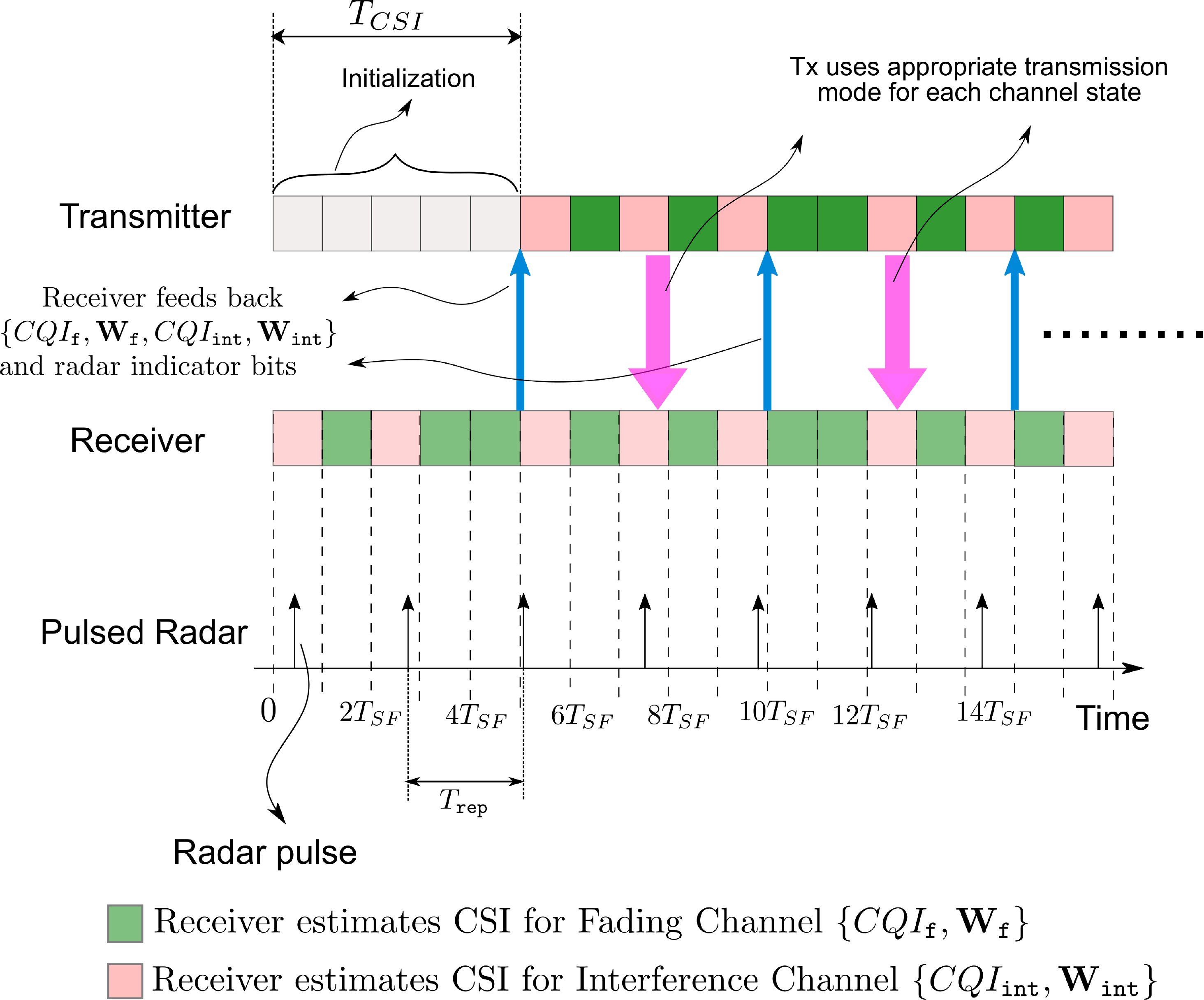}\\
	[-1ex]
	\caption{Illustration of the dual CSI feedback scheme for $T_{CSI} = 5 T_{SF}$, where $T_{SF}$ denotes the duration of each data block. The receiver periodically feeds back the CSI for both channel states $(CQI_\mathtt{f}, \mathbf{W}_\mathtt{f}, CQI_\mathtt{int}, \mathbf{W}_\mathtt{int})$, and the radar indicator bits.}
	\label{Fig_Dual_CSI_Feedback}
\end{figure}

\vspace{-8pt}
\section{Dual CSI Feedback} \label{sec:Dual-CSI-FB}
As discussed in section \ref{Subsec_LinkAdapt_Lim_CSI_FB}, current cellular standards support limited CSI feedback of a single set $CSI = \{ CQI, \mathbf{W} \}$. While this mechanism is efficient in conventional cellular deployments, the presence of pulsed radar interference in a spectrum sharing scenario results in two channel states:
\begin{enumerate}
\item the fading channel, in interference-free data blocks, and
\item the interference-impaired channel, when the pulsed radar is present.
\end{enumerate}
Clearly, a single set of quantized CSI cannot accurately approximate a bimodal channel distribution. In order to handle the additional state in radar-cellular spectrum sharing, we propose \textit{`dual CSI feedback'}, where each user periodically feeds back quantized CSI for both channel states. 
\vspace{-5pt}
\subsection{Feedback Requirements}
In the CSI reporting interval, each user feeds back the set $CSI_\mathtt{dual} = \{ CQI_\mathtt{f}, \mathbf{W}_\mathtt{f}, CQI_\mathtt{int}, \mathbf{W}_\mathtt{int} \}$, where the subscript $\mathtt{f}$ ($\mathtt{int}$) refers to the CSI of the fading (interference-impaired) channel states respectively. 

In addition, the transmitter must know the presence of radar interference in advance, to use the optimal transmission mode for future data blocks. This is enabled by \textit{radar indicator} feedback, which indicates the presence or absence of pulsed radar in each data block, for the \textit{next} $T_{CSI}$ \textit{data blocks}. The receiver can predict the presence of radar interference in a future data block by estimating the $T_\mathtt{rep}$ and monitoring the indices of corrupted OFDM symbols, as discussed in sections \ref{subsec:T_rep_estim} and \ref{subsec:OFDM_index_corrupted}. However, it is worthwhile to note that \textit{radar indicator feedback from a single designated user} is enough for the transmitter to know the indices of future corrupted data blocks. Fig. \ref{Fig_Dual_CSI_Feedback} shows a schematic of the dual CSI feedback scheme, where the \textit{initialization procedure} is used to obtain estimates of $T_\mathtt{rep}$ and $CSI_\mathtt{dual}$ for the first time. 

Assuming a data block duration of $T_{SF} = 1$ ms, if the CSI reporting interval is $T_{CSI}$, then radar indicator feedback consumes $b_\mathtt{rad}$ bits of feedback per CSI reporting interval, where $\lceil \log_2 (T_{CSI}) \rceil \leq b_\mathtt{rad} \leq T_{CSI}$ bits. If the number of active users in the cell is $N_\mathtt{act}$, the total additional feedback overhead is $b_\mathtt{int} = (N_\mathtt{act}N_\mathtt{int} +b_\mathtt{rad})$ bits, where $N_\mathtt{int}$ is the number of additional bits necessary to convey CSI for the interference-impaired fading channel. If $\mathbf{W} \in \mathcal{W}$ and $CQI \in \mathcal{C}$, then $N_\mathtt{int} \geq \big \lceil \log_2 |\mathcal{C}| + \log_2 |\mathcal{W}| \big \rceil$ bits. The corresponding rate overhead is $R_\mathtt{int} = \tfrac{N_\mathtt{act} b_\mathtt{int}}{T_{CSI}}$ bps. 
\vspace{-5pt}

\begin{figure*}[!htbp]
	\centering
	\begin{subfigure}[t]{0.47\textwidth}
		\centering
		\includegraphics[width=3.1in]{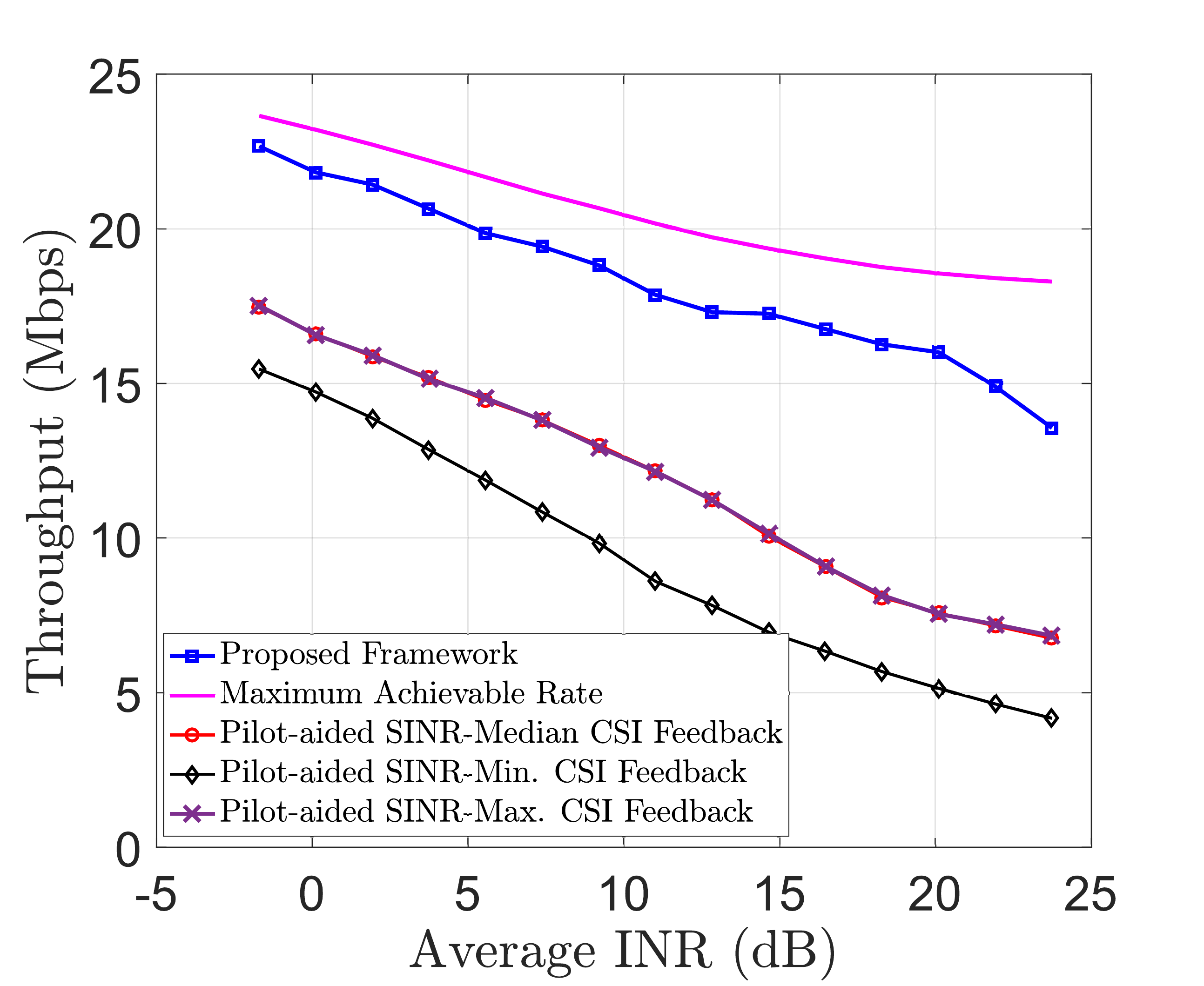}\\
		[-1ex]
		\caption{}
		\label{Fig_TVT_Throughput_Plots}
	\end{subfigure}
	~
	\begin{subfigure}[t]{0.47\textwidth}
		\centering
		\includegraphics[width=3.1in]{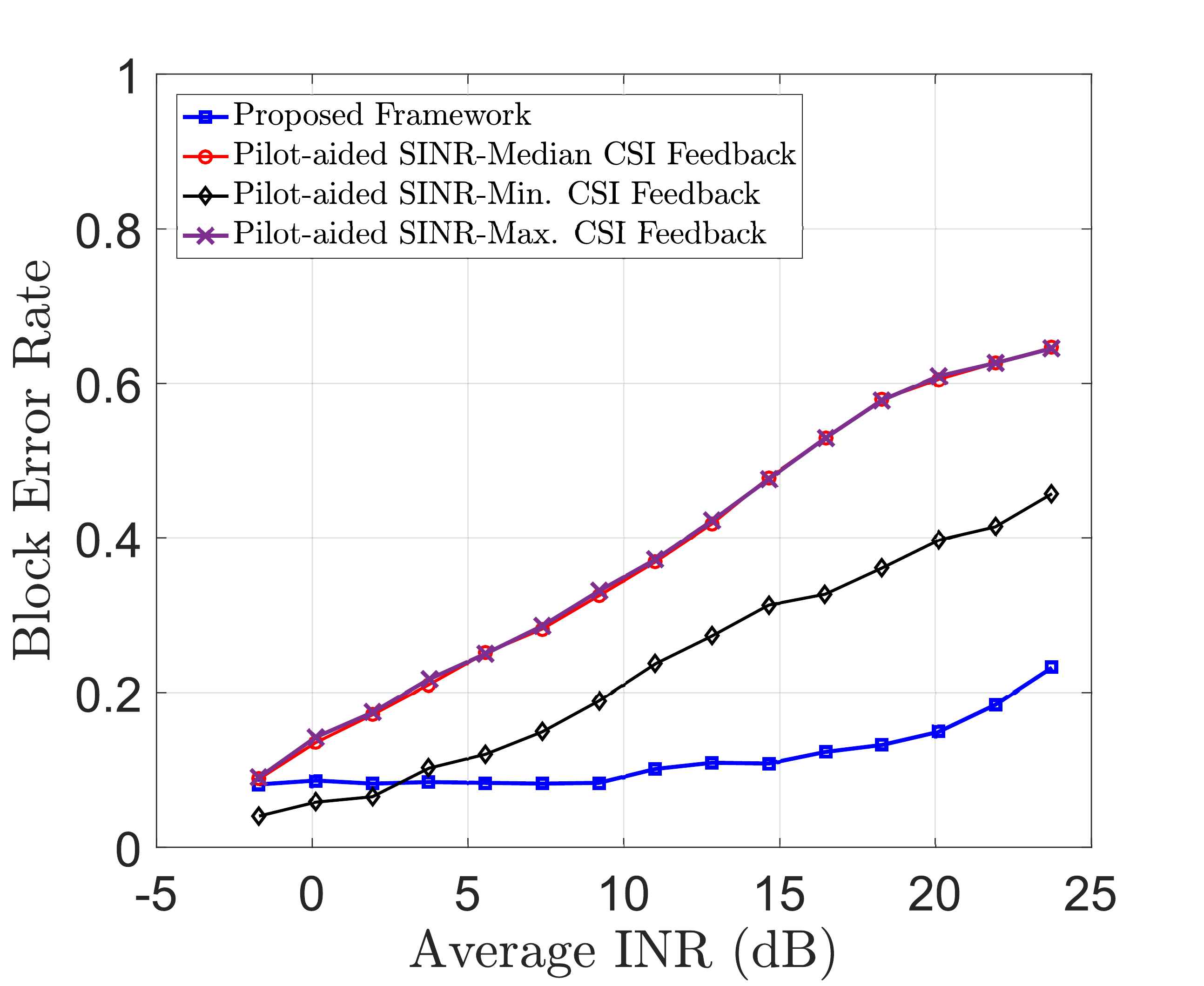}\\
		[-1ex]
		\caption{}
		\label{Fig_TVT_BLER_Plots}
	\end{subfigure}
	~
	\begin{subfigure}[t]{0.47\textwidth}
		\centering
		\includegraphics[width=3.1in]{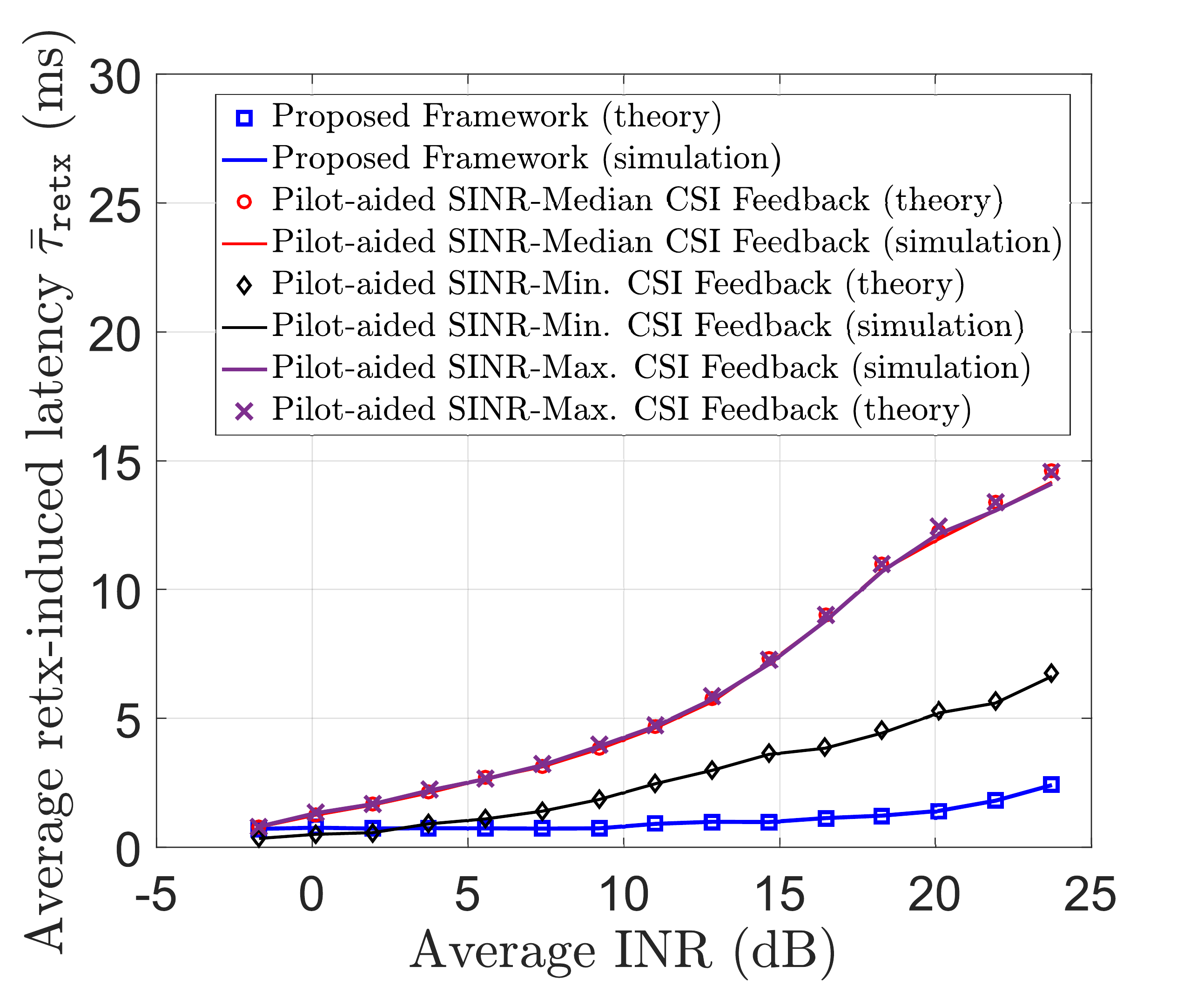}\\
		[-1ex]
		\caption{}
		\label{Fig_TVT_HARQLatency_Plots}
	\end{subfigure}
	\caption{Enhancement of (a) Throughput, (b) Block Error Rate, and (c) Retransmission-induced latency performance, using the proposed hybrid SINR estimation and Dual CSI feedback framework. The average SNR is $19.5$ dB.}
	\label{Fig_LinkAdaptPerf_DualCSIFB}
\end{figure*}

\subsection{Link-Level Performance Improvements}\label{Sec:Num-Results}
In this subsection, we compare the performance of the \textit{hybrid SINR estimation-dual CSI feedback} framework with the \textit{pilot-aided SINR estimation-single CSI feedback} scheme (henceforth referred to as the `conventional' scheme). We developed a 3GPP-compliant link-level simulator to analyze radar-LTE coexistence scenarios, using the MATLAB LTE/NR toolboxes\textsuperscript{TM}, using the system parameters shown in Table \ref{Table_Radar_LTE_coexist_HybSINR_DCFB}. For the conventional scheme, we consider (a) mininum CSI feedback, (b) median CSI feedback, and (c) and maximum CSI feedback schemes that were described in section \ref{Subsec_LinkAdapt_Lim_CSI_FB}.

Fig. \ref{Fig_LinkAdaptPerf_DualCSIFB} compares the link-level performance of the proposed framework with the conventional scheme. Fig. \ref{Fig_TVT_Throughput_Plots} shows the throughput as a function of the average INR when the average $\text{SNR} = 19.5$ dB. We observe that the proposed framework achieves a $30\%-100 \%$ rate enhancement when compared to median and maximum CSI feedback, and a $47\%-225 \%$ rate enhancement compared to minimum CSI feedback. In addition, we also observe that our framework achieves $74\%-96\%$ of the maximum achievable rate over a wide range of INR values, demonstrating a high utilization of the channel capacity. 

It is important to observe that the rate improvement due to the proposed framework balances the BLER constraints as shown in Fig. \ref{Fig_TVT_BLER_Plots}, where $\text{BLER} \leq 0.1$ for $\text{INR} \leq 12$ dB. Interestingly, the BLER performance at high INR improves significantly when compared to minimum CSI feedback, the most conservative conventional scheme. As expected, median and maximum CSI feedback always result in a higher BLER compared to minimum CSI feedback. This is because it requires a higher number of interference-impaired pilots per estimation window to mimic the performance of minimum CSI feedback. 

High BLER due to decoding failures result in degradation of the \textit{HARQ-induced latency}, which is defined as the latency due to hybrid ARQ (HARQ) retransmissions in LTE and NR. The average HARQ-induced latency $(\bar{\tau}_\mathtt{retx})$ is approximately given by \cite{Rao_VTC_LTE_LinkAdapt_2019}
\begin{align}
\label{retx_induced_latency}
\bar{\tau}_\mathtt{retx} = \frac{\mathtt{BLER} \times \bar{\tau}_\mathtt{wait}}{1 - \mathtt{BLER}}.
\end{align}
$\bar{\tau}_\mathtt{wait}$ is the average wait time between consecutive retransmissions. We assume $\bar{\tau}_\mathtt{wait} = 8$ ms, which is the typical value in LTE and NR \cite{sesia2011lte}, \cite{dahlman20185g}. Fig. \ref{Fig_TVT_HARQLatency_Plots} shows that the proposed framework improves retransmission induced latency by a factor of $3$ when compared to minimum CSI feedback, and by an order of magnitude when compared to median CSI feedback. In addition, we observe that theoretical and simulation values are in good agreement.

In LTE and NR, CSI feedback for single-user transmission modes has a overhead of about $b_\mathtt{fb}=10$ bits per CSI estimation interval $T_{CSI}$, where $T_{CSI} \geq 2$ ms \cite{dahlman20185g}, \cite{sesia2011lte}. Therefore, in a cell with $N_\mathtt{act} = 100$ active users, the additional rate overhead due to dual CSI feedback will satisfy $r_\mathtt{int} \leq \tfrac{100 \times (10)}{2 \times 10^{-3}} + \tfrac{1}{1 \times 10^{-3}} = 510$ kbps. 

In summary, the proposed framework simultaneously improves throughput, BLER, and latency performance when compared to conventional schemes in the presence of pulsed radar interference. For most operational regimes, the downlink throughput improvement is significantly high to justify the use of dual CSI feedback. For MU-MIMO transmission modes in NR that typically need $100$ bits/user/CSI estimation interval \cite{dahlman20185g}, further investigation is needed to evaluate the performance achieved using our framework. In general, dual CSI feedback is beneficial if the \textit{cell-wide throughput gain} is greater than the additional uplink rate overhead.%

\vspace{-10pt}

\section{Conclusion}\label{sec:Conclusions}
In this paper, we developed a comprehensive semi-blind SINR estimation framework using pilot-aided and heuristic-aided estimates to compute the wideband post-equalizer SINR in radar-cellular coexistence scenarios. We characterized the distribution of a low complexity max-min heuristic under a tractable signal model, and demonstrated its accuracy and robustness for interference-impaired QAM data symbols. To handle channel bimodality due to periodic transitions between the \textit{fading} and the \textit{interference-impaired} channel states, we proposed a \textit{dual CSI feedback} mechanism where the receiver reports quantized CSI for both channel states. Unifying these two schemes and using radar-LTE-A Pro coexistence as an example, we demonstrated significant improvements in key link-level performance metrics such as throughput, BLER and retransmission-induced latency \textit{simultaneously}. 

In vehicular communication systems such as C-V2X, link adaptation decisions need to be taken at a faster timescale due to the highly dynamic wireless channel. Co-channel or adjacent channel pulsed radar interference inhibits accurate CSI acquisition, which adversely impacts the rate and latency performance of a vehicular link. The semi-blind SINR estimation and dual CSI feedback framework proposed in this paper addresses the issue of accurate CSI acquisition in the presence of such wideband intermittent interference signals. Further, the low computational complexity and low overhead of the proposed framework promises a high potential for being effective in dynamic channel conditions, and hence is attractive for implementation in vehicular communication systems sharing spectrum with a high-powered radar.

Investigation of the optimal SU- and MU-MIMO precoder estimation in non-pilot interference is a useful extension to this work, which is especially important in multi-antenna transmission modes of LTE-A Pro and 5G NR. In addition, novel scheduling and resource management schemes based on this framework can also be developed for different applications such as vehicular-to-everything (V2X) and Internet of Things (IoT) services coexisting with radar. Such scenario-specific frameworks will be of practical importance to enable efficient link adaptation mechanisms in radar-5G/6G coexistence since rate and latency performance often need to be jointly optimized in these scenarios. 
\vspace{-10pt}
\appendix
\section*{Proof of Lemma \ref{Lemma_Cond_Dist_D_min_X}} \label{Appndx_A}
The conditional CDF of $\{D|X,\Phi\}$ can be written as 
\begin{align*}
F_{D}(d|x,\phi) = \int_{ \mathcal{A}_{n_R}} \int_{\mathcal{A}_{n_I}} \mathbb{P}[D \leq d|x, \phi, n] f_{N} (n) d n,
\end{align*}
where $f_N(n) = f_{N_R}(n_R) f_{N_I}(n_I)$. By equation (\ref{Expanding_gen_form_of_Dmin}) we can observe that $D$ is a Rician random variable, since it is the amplitude of a complex Gaussian where the real/imaginary parts have a different mean. Thus, the integral can be transformed into polar coordinates $(z,\theta)$ to get an integral of the form 
\begin{align}
\label{Cond_CDF_defn_D_polar}
F_{D}(d|x,\phi) = \int_{0}^{d} \int_{ \mathcal{A}_{\theta}(z)} f_{Z,\Theta} (z,\theta|x,\phi) d \theta d z, d \geq 0,
\end{align}
where $f_{Z,\Theta} (z,\theta|x,\phi)$ is the conditional density function of $\{ Z, \Theta\}$. Depending on the value of $D$, there are 3 distinct regions of integration for QAM constellations: (a) $0 \leq d \leq \tfrac{d_\mathtt{c}}{2}\ \forall\ x \in \mathcal{X}$, (b) $\tfrac{d_\mathtt{c}}{2} \leq d \leq \tfrac{d_\mathtt{c}}{\sqrt{2}}\ \forall\ x \in \mathcal{X}$, and (c) $\frac{d_\mathtt{c}}{\sqrt{2}} \leq d \leq +\infty$ for $x \in \mathcal{X}_\mathtt{bnd}$. Fig. \ref{Fig3_AllCases} shows these regions for 16-QAM. We denote the distance of each point $x_i \in \mathcal{X}$ along the x- and y-axes to the edges of its decision region is given by $d^{L,R}_{x_i}, d^{U,R}_{x_i}$ and $d^{L,I}_{x_i}, d^{U,I}_{x_i}$ respectively. Table \ref{Table_dist_to_decsn_bdries_16QAM} shows these boundaries for points in the first quadrant of a 16-QAM constellation. Below, we derive the conditional distribution of $\{D|X, \Phi\}$ for each region, by leveraging the properties of Rician r.v's. 

\begin{figure}[!t]
	\centering
	\includegraphics[width=3.5in]{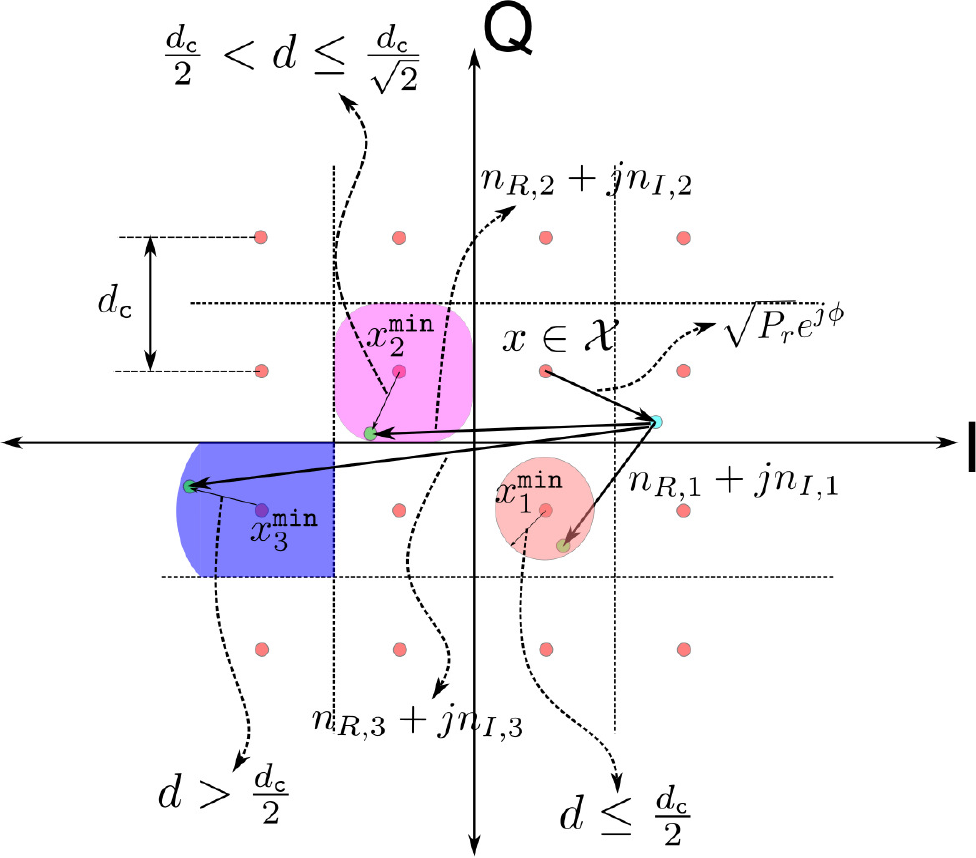}
	\caption{Illustration of the different range of values for $d$, and the corresponding region of integration to derive $F_D(d)$. Cases 1 and 2 are possible for all $x \in \mathcal{X}$, but case 3 is possible only for $x \in \mathcal{X}_\mathtt{bdry}$. }
	\label{Fig3_AllCases}
\end{figure}

\vspace{-7pt}
\subsection{Case 1: $0 \leq d \leq \tfrac{d_\mathtt{c}}{2}\ \forall\ x^{(j)} \in \mathcal{X}$}
In this case, the region of integration is \textit{circular} with radius $d$ as shown in Fig. \ref{Fig3_AllCases} (red shaded region). Defining $m_{R,j} \triangleq [x_{R} - x^{(j)}_{R} + \sqrt{P_r} \cos(\phi) + n_R], m_{I,j} \triangleq [x_I - x^{(j)}_{I} + \sqrt{P_r} \sin(\phi) + n_I]$, and $\theta \triangleq \tan^{-1} \big( \tfrac{m_{I,j}}{m_{R,j}} \big)$ in equation (\ref{Expanding_gen_form_of_Dmin}), and conditioning on $x, x^{(j)} \in \mathcal{X}$ and $\phi \in [0,2\pi]$, we observe that when the nearest neighbor is $x^{(j)} \in \mathcal{X}$, 
$\{D_\mathtt{min}| X, \phi\} \sim \mathtt{Rician} (\nu_j, \sigma^2)$ with parameters $\nu^2_j = m^2_{R,j} + m^2_{I,j}$ and $\sigma^2 = \tfrac{\sigma^2_n}{2}$. Therefore, we have 
\begin{align}
\label{Case1_FullCond}
f_{Z}(z|x,\phi) & = \sum\nolimits_{x^{(j)} \in \mathcal{X}} \tfrac{2z}{\sigma^2_n} e^{-\frac{z^2 + \nu^2_j}{\sigma^2_n}} I_0 \big(\tfrac{2 \nu_j d}{\sigma^2_n} \big), z \geq 0, \text{ and } \nonumber \\
F_{D} (d|x, \phi) & = \sum\nolimits_{x^{(j)} \in \mathcal{X}} \int_{0}^{d} \tfrac{2z}{\sigma^2_n} e^{-\frac{z^2 + \nu^2_j}{\sigma^2_n}} I_0 \big(\tfrac{2 \nu_j d}{\sigma^2_n} \big)d z, \nonumber \\
& \stackrel{(a)}{=} \sum\nolimits_{x^{(j)} \in \mathcal{X}} \Big[ 1 - Q_1 \Big( \tfrac{\sqrt{2} \nu_j}{\sigma_n}, \tfrac{\sqrt{2} d}{\sigma_n} \Big) \Big],
\end{align}
where $I_0 (\cdot)$ is the Bessel function of the first kind with order zero, and (a) is obtained by simplifying the CDF of a Rician random variable in the form of a Marcum Q-function with parameters $(M, a, b)$ \cite{abramowitz2012handbook}. 
 
\begin{figure}[t]
	\centering
	\includegraphics[width=3.5in]{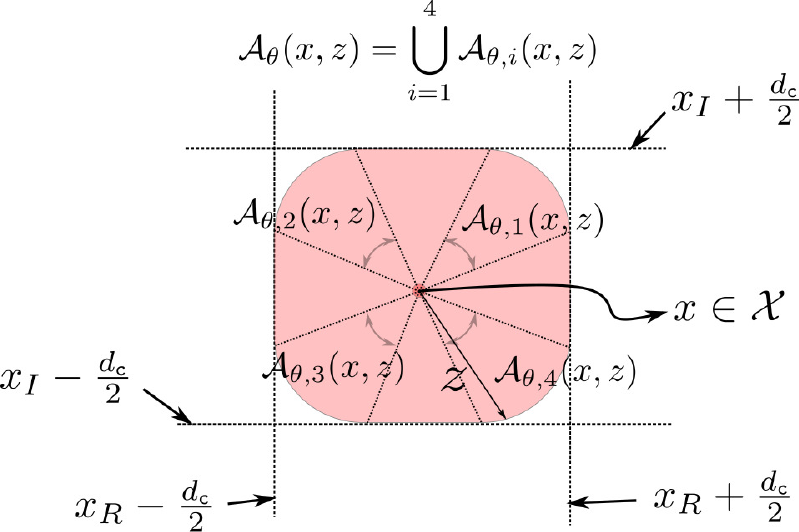}\\
	[-2ex]
	\caption{Illustration of the limits of $\theta$ in equation (\ref{Theta_region_Case2_1}) when $x \in \mathcal{X}_\mathtt{int}$.}
	\label{Fig4_Case2_illustration}
\end{figure}

\subsection{Case 2: $\tfrac{d_\mathtt{c}}{2} \leq d \leq \tfrac{d_\mathtt{c}}{\sqrt{2}}\ \forall\ x^{(j)} \in \mathcal{X}$}
In this case, the integration region for each point $x_j$ is a \textit{`truncated' circle}, as shown in Fig. \ref{Fig3_AllCases} (pink colored region). The minimum distance $D$ is a Rician random variable with a \textit{radially asymmetric integration region}. Therefore, the conditional density in polar coordinates is given by
\begin{align}
\label{Rician_z_theta_joint_PDF_1}
f_{Z, \Theta}(z, \theta|x, \phi) = \sum\nolimits_{x^{(j)} \in \mathcal{X}} \tfrac{z}{\pi \sigma^2_n} e^{-\frac{z^2 + \nu^2_j + 2z m_j}{\sigma^2_n}}, 
\end{align}
for $z \geq 0, 0 \leq \theta \leq 2 \pi$. The region of integration of $\Theta$ for $x \in \mathcal{X}$ is a function of $z$, given by
\begin{align*}
\mathcal{A}_{\theta}(x, z) = & \big\{ \theta \big| d^{L,R}_{x} \leq z \cos \theta \leq d^{U,R}_{x}, d^{L,I}_{x} \leq z \sin \theta \leq d^{U,I}_{x}  \big\}.
\end{align*}
For $x \in \mathcal{X}_\mathtt{int}$, the above can be simplified as 
\begin{align}
\label{Theta_region_Case2_1}
\mathcal{A}_{\theta}(x, z) =& \bigcup_{i=1}^4 \mathcal{A}_{\theta,i} (x, z), \text{ where } \nonumber \\
\mathcal{A}_{\theta,i} (x, z) =& \Big \{\theta \big| (i-1)\tfrac{\pi}{2} + \cos^{-1} \big( \tfrac{d_\mathtt{c}}{2z} \big) \leq \theta \leq (i-1)\tfrac{\pi}{2} + \sin^{-1} \big( \tfrac{d_\mathtt{c}}{2z} \big) \Big \}.
\end{align}
Fig. \ref{Fig4_Case2_illustration} shows an example of the integration region for $x^{(j)} \in \mathcal{X}_\mathtt{int}$. Using (\ref{Rician_z_theta_joint_PDF_1})-(\ref{Theta_region_Case2_1}) and marginalizing $\Theta$ and $X$, we obtain the desired result. 

\subsection{Cases 3 and 4: $d \geq \tfrac{d_\mathtt{c}}{\sqrt{2}}$ for all $x^{(j)} \in \mathcal{X}$}
Derivation of the conditional CDF is similar to that in Case 2. The additional constraint here is that $F_{D |X , \Phi} (d |x , \phi)$ is non-zero iff $x^{(j)} \in \mathcal{X}_\mathtt{bnd}$. This is because for interior points, $0 \leq D_\mathtt{min} \leq \tfrac{d_\mathtt{c}}{\sqrt{2}}$ is always true for QAM modulation schemes, as illustrated in Fig. \ref{Fig3_AllCases}.
\bibliographystyle{IEEEtran}
\bibliography{TVT_2019_references}
\end{document}